\documentclass[10pt, conference, letterpaper]{IEEEtran}

\usepackage{cite}
\usepackage{amsmath,amssymb,amsfonts}
\usepackage{algorithmicx}
\usepackage{array}
\ifCLASSOPTIONcompsoc
  \usepackage[caption=false,font=normalsize,labelfont=sf,textfont=sf]{subfig}
\else
  \usepackage[caption=false,font=footnotesize]{subfig}
\fi

\usepackage{url}
\usepackage{graphicx}
\usepackage[table,xcdraw]{xcolor}
\usepackage{amsthm}
\usepackage{svg}
\usepackage{bm}
\usepackage{algorithm}
\usepackage[noend]{algpseudocode}
\usepackage{setspace}

\usepackage[OT1,T1]{fontenc}
\usepackage[scaled=.88]{berasans}

\DeclareMathOperator*{\argmax}{argmax}

\algdef{SE}[DOWHILE]{Do}{doWhile}{\algorithmicdo}[1]{\algorithmicwhile\ #1}%
\newtheorem{theorem}{Theorem}

\algrenewcommand\algorithmicindent{1em}%

\begin{document}

\title{Multi-Provider NFV Network Service Delegation via Average Reward Reinforcement Learning}

\author{
\IEEEauthorblockN{
	Bahador Bakhshi,
	Josep Mangues-Bafalluy, 
	Jorge Baranda
}
\IEEEauthorblockA{Centre Tecnologic de Telecomunicacions de Catalunya (CTTC), Spain\\
\{bbakhshi, jmangues, jbaranda\}@cttc.cat}
}
\maketitle

\begin{abstract}
In multi-provider 5G/6G networks, service delegation enables administrative domains to federate in provisioning NFV network services. Admission control is fundamental in selecting the appropriate deployment domain to maximize average profit without prior knowledge of service requests’ statistical distributions. This paper analyzes a general federation contract model for service delegation in various ways. First, under the assumption of known system dynamics, we obtain the \emph{theoretically} optimal performance bound by formulating the admission control problem as an infinite-horizon Markov decision process (MDP) and solving it through dynamic programming. Second, we apply reinforcement learning to \emph{practically} tackle the problem when the arrival and departure rates are not known. As Q-learning maximizes the discounted rewards, we prove it is not an efficient solution due to its sensitivity to the discount factor. Then, we propose the \emph{average reward} reinforcement learning approach (R-Learning) to find the policy that directly maximizes the average profit. Finally, we evaluate different solutions through extensive simulations and experimentally using the 5Growth platform. Results confirm that the proposed R-Learning solution always outperforms Q-Learning and the greedy policies. Furthermore, while there is at most 9\% optimality gap in the ideal simulation environment, it competes with the MDP solution in the experimental assessment.
\end{abstract}


\IEEEpeerreviewmaketitle

\section{Introduction}
\label{sec:Intro}

Service provisioning in 5G/6G networks is challenging in a context with diverse quality of service (QoS) requirements, heterogeneity of infrastructure resources and shrinking per-service revenues.
To cope with these challenges, innovative principles including network slicing, network softwarization using Software Defined Networking (SDN) and Network Function Virtualization (NFV), and multi-domain service orchestration have been proposed in the architecture of the networks
\cite{alliance20175g,virtualisation2018release, ITU2020Net2030}. In multi-domain service orchestration, multiple providers federate/collaborate in provisioning network services (NSs) consisting of Virtual Network Functions (VNFs) interconnected by virtual links
\cite{rosa2015md2, bhamare2017optimal, hortiguela2020realizing}. In \emph{multi-provider service delegation}, as a kind of multi-domain orchestration, the customer of the consumer domain (CD) requests an NS that is either deployed locally in the CD or delegated to the peering provider domain (PD). This is done transparently to the customer (i.e., CD decides) based on a CD-PD federation contract specifying the technical and business agreements between the domains \cite{valcarenghi2018framework}.
In this paper, we use multi-domain and multi-provider interchangeably.

Admission control (AC) is a key issue in multi-provider service delegation as it determines the service deployment domain \cite{antevski2020q}. Indeed, the AC of the CD 
makes the highest level of service-orchestration decisions, which directly impacts on the business profit of the CD domain. It decides how the local resources of the CD and the external resources, offered by the PD through the established federation contract, should be used for provisioning heterogeneous services with different revenues. 
The admission controller decisions get more important as the \emph{heterogeneity} of services increases, which is the case in 5G/6G networks, and consequently, different demands should be treated in different ways. Inappropriate decisions by the admission controller, that do not take the heterogeneity into account, will lead to QoS degradation and profit loss. Therefore, AC is a powerful tool for the service provider to aim for various objectives, like profit maximization, load balancing, or QoS guarantees \cite{baranda2020nfv}, by deploying each NS in the most appropriate domain. 

In this paper, we study the AC of the Multi-Provider Service Delegation (AC-MPSD) problem, where there is a CD and a PD. The federation contract established between domains specifies the service catalog they share, the resource quotas reserved in the PD for delegation purposes, and the per-service associated cost. Upon arrival of an NS request, without prior knowledge of future NS requests, the AC decides either to admit or to reject the request. In case of admission, it also determines the deployment domain. In this problem, the objective is to find an AC policy that maximizes the long-term average profit of the CD subject to the delegation cost. This problem is not fully addressed in the literature. In the last years, various approaches have been proposed for AC in 5G \cite{ojijo2020survey}; however, they cannot be applied directly to the AC-MPSD problem as they consider single domain networks.

Recently, AI/ML techniques have been extensively applied to networking problems \cite{xie2018survey, morocho2019machine, bega2019machine}. Approaching the AC-MPSD problem with AI/ML methods is also promising. In this problem, the admission controller decides for each NS request without prior knowledge of future demands, hence it is an instance of the \textit{sequential decision making under uncertainty} problem, which can be tackled  efficiently by Reinforcement Learning (RL)  \cite{sutton2018reinforcement}. While finding an optimal AC policy using RL has been studied previously \cite{wu2001admission, han2018markov, raza2018slice}, they cannot be applied to the AC-MPSD problem, as they were proposed in contexts other than multi-provider service delegation. 

Existing AC solutions
were mainly evaluated in simulation environments that raise questions about their performance and efficiency in real networks. In this paper, besides in extensive simulations, the proposed solution is evaluated in a realistic environment using the publicly available 5Growth platform
\footnote{https://github.com/5growth}~\cite{Xi_commag_2021}, an NFV/SDN-based orchestration framework with AI/ML capabilities to perform closed-loop automation and zero-touch service and network management. 

To sum up, to the best of our knowledge, despite the potential of the service delegation concept in the multi-provider 5G/6G networks and the crucial role of AC in this context, the AC-MPSD problem has not yet been studied in detail. In this paper,we extend our two preliminary works on this problem \cite{antevski2020q, bakhshi2021rl} and 
we make the following contributions:
\begin{itemize}
\item Using a flexible federation contract model, and by assuming known system dynamics, the problem is formulated as a Markov Decision Process (MDP). Its solution by the Dynamic Programming (DP) Policy Iteration (PI) algorithm provides the theoretical optimal AC policy.
\item The drawback of applying the widely-used Q-Learning algorithm to the AC-MPSD problem, due to the sensitivity to the discount factor is analytically proved and also numerically justified. 
\item A model-free \emph{average reward}-based RL algorithm is proposed as a practical solution to maximize the long-term average profit.
\item The proposed solution is implemented and evaluated in an experimental setup using the 5Growth orchestration platform as well as with extensive simulations that show a near-optimal performance.
\end{itemize}

The remainder of this paper is organized as follows. In Section \ref{sec:Related}, the related works are reviewed. The system model and problem statement are discussed in Section \ref{sec:Model_Problem}. In Section \ref{sec:Optimal}, the problem is formulated as an MDP and solved by PI. The RL approaches are discussed in Section \ref{sec:Reinforce}. In Section \ref{sec:Simulation}, the numerical results of  the simulations, as well as the experimental testbed results, are presented. Finally, in Section \ref{sec:Conclusion}, we conclude this paper.

\section{Related Work}
\label{sec:Related}
In this section, we review three categories of related work, i.e., ($i$) service federation in multi-domain networks, ($ii$) AC in 5G networks, and ($iii$)  RL-based AC solutions; and identify the differences between those studies and this paper.

Multi-domain orchestration is an inherent concept in 5G/6G network architecture \cite{alliance20175g, virtualisation2018release}, but its realization needs resolving research challenges and implementation issues \cite{rosa2015md2, hortiguela2020realizing}. From the theoretical point of view, the problem is formulated as optimization models, and due to the complexity of the problem, heuristic algorithms are proposed to find sub-optimal solutions in \cite{bhamare2017optimal, dietrich2017multi, sun2018service}. These preliminary works were extended later to consider more complex objective functions, e.g., energy efficiency \cite{sun2019energy, kaur2019energy}, and network service latency  \cite{zhang2019cost}. 
To tackle the complexity of the problem, a topology aggregation technique \cite{yan2020service} and a deep learning-based solution \cite{zhang2021intelligent} were also proposed recently.  From the implementation point of view, the architectural framework for multi-domain service orchestration \cite{toumi2021cross} and, more specifically, service federation \cite{valcarenghi2018framework, li2018service} were also studied. In the 5G-Transformer platform \cite{5GTrans, baranda20205g}, the service federation component was developed so that it is capable of deploying NSs spanning multiple domains and transport networks. While these theoretical studies and the practical developments address some aspects of the multi-domain service orchestration, they do not specifically consider the AC problem; i.e., they assume that the service has already been accepted and attempt to efficiently deploy it.

Admission control in 5G networks has been the topic of many studies in recent years \cite{ojijo2020survey}. Various objectives are aimed, including revenue maximization \cite{nejad2018vspace, bega2019machine} and fairness assurance \cite{han2019utility}. Different strategies, and techniques have been utilized to achieve these goals. The most straightforward approach is to \emph{greedily} attempt to accept any given NS request. However, to get closer to the optimal policy, techniques based on optimization theory \cite{nejad2018vspace}, and reinforcement learning \cite{bega2019machine} have also been proposed. In \cite{bega2017optimising}, the authors formulated the AC problem in the case of two inelastic and elastic traffic models as a Semi-Markov decision problem, and then obtained the optimal policy to maximize the revenue of the service provider. These solutions cannot be applied directly to the AC-MPSD problem, since they are proposed for single-domain networks, hence without service delegation.

RL is an adequate tool to deal with the AC problem through which the admission controller learns the appropriate policy via the rewards gained over time. For the first time, AC in multimedia networks was approached by RL in \cite{tong2000adaptive}. Later the authors extended the problem and dealt with the joint routing and admission control via RL in \cite{tong2002reinforcement}. 
Recently, the joint AC and routing problem in SDN has been investigated via approximate dynamic programming \cite{yang2016joint}. In wireless networks, RL-based AC solutions have also been proposed. In \cite{wu2001admission}, AC in cellular networks is formulated as an MDP. AC in CDMA networks using RL was studied in \cite{liu2005self}. Recently, in 5G networks, the network slice admission control problem is formulated as an MDP in \cite{han2018markov}, and RL-based solutions are proposed in \cite{raza2018slice, caballero2018network}. While these works approach the AC problem using RL, they are not applicable to the service delegation problem, since the contexts of those problems are quite different from multi-domain service orchestration.

In the most closely related work, Q-Learning was applied to a similar problem \cite{antevski2020q}. However, this paper differs from that one by considering a more general and flexible federation contract model, analytically investigating the limitation of Q-Learning, proposing an average reward-based RL algorithm, and implementing the solution using the 5Growth platform.

\section{System Model and Problem Statement}
\label{sec:Model_Problem}

\subsection{Assumptions and System Model}
\label{sec:Assumptions}
In this paper, we consider a multi-provider network, which is composed of a CD and a PD\footnote{The presented analyses and solutions in this paper can be extended to the scenarios with multiple PDs without substantial modifications; it is omitted for the sake of simplicity of presentation.}. Via the federation contract established between the domains, the CD, in addition to its own local resources, uses the reserved resource quotas in the PD to satisfy service provisioning requests. In these domains, there are $\mathcal{R}=\{1,2,\ldots, R\}$ types of resources. The total resources of the CD are denoted by vector $\bar{\bm{C}}^{l}=[\bar{C}^{l}_{1}, \ldots, \bar{C}^{l}_{R}]$ where $\bar{C}^{l}_{k}$ is the total amount of resource type $k$ in this domain; e.g., the total number of CPU cores. According to the federation contract, the PD also reserves $\bar{\bm{C}}^{p}=[\bar{C}^{p}_{1}, \ldots, \bar{C}^{p}_{R}]$ amount of resources for the delegated NSs. Define $\bm{C}^{l}$ and $\bm{C}^{p}$ as, respectively, the current available capacities of $\bar{\bm{C}}^{l}$ and $\bar{\bm{C}}^{p}$.

$\mathcal{I}=\{1,2,\ldots,I\}$ types of NSs are requested to the CD; each type $i \in \mathcal{I}$ is defined by tuple $(\bm{c}_{i},r_{i})$ in the service catalog, where $r_{i}$ is the revenue 
of the CD if an instance of the service is admitted; and $\bm{c}_{i} = [c_{i,1}, \ldots, c_{i,R}]$, where  $c_{i,k}$ is the total aggregated amount of resources type $k$ required by the VNFs of the NS, e.g., the total required CPU cores by all the VNFs. We assume that the arrival of the requests for NS type $i$ as well as the departure of those NSs are Poisson processes with average rate $\lambda_{i}$ and $\mu_{i}$ respectively. Therefore, an NS $\delta_{i}$ of type $i$ as far as timing is concerned is specified by $(\tau^{s}_{\delta}, \tau^{e}_{\delta}, i)$ where  $\tau^{s}_{\delta}$ and $\tau^{e}_{\delta}$ are respectively the arrival and departure time of the NS request that are determined by $\lambda_{i}$ and $\mu_{i}$, and $i$ is the type of the requested service.

In the considered model of the federation contract, even if the reserved quota $\bar{\bm{C}}^p$ in the PD is exceeded, NS requests can still be delegated but at an additional overcharged cost until the resource consumption by the delegated NSs exceeds a reject threshold. More specifically, in the federation contract, besides the $\bar{\bm{C}}^p$, three additional parameters are specified: ($i$) delegation fees $\bm{\Sigma} = [\sigma_{1}, \ldots, \sigma_{I}]$, ($ii$) overcharging scales $\bm{\Omega} = [\omega_{1}, \ldots, \omega_{I}] \geq \bm{1}$, and ($iii$) reject thresholds $\bm{\Theta} = [\theta_{1}, \ldots, \theta_{R}] \geq \bm{1}$. 
We define $\bar{\bm{C}}^{p}_{\theta} = \bm{\Theta}\otimes\bar{\bm{C}}^{p}$, which is the element-wise multiplication of vectors $\bm{\Theta}$ and $\bar{\bm{C}}^{p}$, and let ${\bm{C}}_{\theta}^p$ be the current available capacity of $\bm{\Theta}\otimes\bar{\bm{C}}^{p}$. Based on these parameters, the cost of delegating $\delta_{i}$ is as follows where the vectors are compared element-wise:
\begin{equation*}
\label{eq:Delta}
\Delta(\delta_{i}) = 
	\begin{cases}
	\sigma_{i}, & \text{if } \bm{c}_{i} \leq \bm{C}^{p} \\
	\omega_{i} \sigma_{i}, & \text{if } \exists \, k \text{ s.t. }  {c}_{i,k} > {C}^{p}_{k} \text{ and } \bm{c}_{i} \leq \bm{C}^{p}_{\theta}
	\end{cases}
\end{equation*}
Note that in the case of ${c}_{i,k} > {C}^{p}_{\theta,k}$ for some $k$, the NS cannot be deployed in the PD. Moreover, we assume that the delegation cost of an NS $\delta_{i}$ is determined by $\Delta(\delta_{i})$ at the arrival time $\tau^{s}_{\delta}$ and does not change later. Finally, in general, $\omega_{i}$ can be a function of $\bm{C}^{p}$, $\bm{C}^{p}_{\theta}$, $\bm{c}_{i}$ without affecting the problem formulation and the proposed solutions.

It is worth noting that this federation contract model has the flexibility to implement various pricing strategies, for example by setting $\bm{\Theta} = \bm{1}$, there won't be any overcharged request, i.e., $\delta_{i}$ will be rejected if $\exists \, k $ s.t. ${c}_{i,k} > {C}^{p}_{k}$; or if $\bm{\Theta}$ is sufficiently large, there won't be any rejection, i.e., PD always accepts delegated NSs but overcharges them.

\subsection{Problem Statement}
\label{sec:Problem}
In this paper, we study the following on-line AC-MPSD problem. There is a CD with capacity $\bar{\bm{C}}^l$ that provides $\mathcal{I}$ types of services. It established a federation contract $(\bar{\bm{C}}^{p}, \bm{\Sigma}, \bm{\Omega}, \bm{\Theta})$ with a PD. NS requests for each type $i$ arrive one-by-one at rate $\lambda_{i}$. Upon such arrival, the admission controller, without knowledge of the future requests, decides whether to ($i$) accept the NS request to be deployed in the CD, yielding profit $r_{i}$, which is possible only if $\bm{c}_{i} \leq {\bm{C}}^{l}$, or ($ii$) delegate it to the PD, only if $\bm{c}_{i} \leq {\bm{C}}_{\theta}^{p}$, that yields profit $r_{i} - \Delta(\delta_{i})$, or ($iii$) reject the NS request. The admitted NSs will depart the network at rate $\mu_{i}$ (defined per service type).

Define $\mathcal{D}_{t} = \{\delta_{i} \ \forall i \in \mathcal{I} \text{ s.t. } \tau^{s}_{\delta} \leq t\}$ as the set of NS requests arrived before time $t$. Let $\mathcal{L}_{t}$ be the set of the NSs $\delta_{i} \in \mathcal{D}_{t}$ deployed locally in the CD, and similarly $\mathcal{F}_{t}$ is the set of   delegated NSs to the PD. The AC-MPSD problem is
\begin{equation}
\label{eq:Profit}
\text{max }\lim_{t \rightarrow \infty} \frac{1}{|\mathcal{D}_{t}|}\sum_{i \in \mathcal{I}} \Big(\sum_{\delta_{i} \in \mathcal{L}_{t}} r_{i} + \sum_{\delta_{i} \in \mathcal{F}_{t}} \big(r_{i} - \Delta(\delta_{i})\big)\Big),
\end{equation}
subject to:

{\setlength{\abovedisplayshortskip}{-3.5mm}
\begin{equation}
\label{eq:ConsumerCapacity}
\sum_{i \in \mathcal{I}} \sum_{\delta_{i} \in \mathcal{L}_{t}} \bm{c}_{i} \leq \bar{\bm{C}}^{l}, \ \   \forall t, 
\end{equation}}
\begin{equation}
\label{eq:ProviderCapacity}
\sum_{i \in \mathcal{I}} \sum_{\delta_{i} \in \mathcal{F}_{t}} \bm{c}_{i} \leq \bar{\bm{C}}^{p}_{\theta},  \ \ \forall t,
\end{equation}
where \eqref{eq:Profit} is the long-term average profit of the CD; and \eqref{eq:ConsumerCapacity}
and \eqref{eq:ProviderCapacity} respectively satisfy the capacity constraint of the CD and PD. This problem cannot be solved by the traditional optimization theory techniques, e.g., integer programming, because the NS requests arrive one-by-one and all the required information is not available at the beginning. 

\section{Optimal Solution}
\label{sec:Optimal}
The AC-MPSD is an instance of the \textit{sequential decision making under uncertainty} problem where a decision-making agent takes a sequence of decisions in an uncertain environment. Each decision, besides the uncertain dynamics of the environment, changes the state of the environment and leads to a reward. The decision maker's objective is to maximize a cumulative long-term reward. Assuming that the dynamics of the environment are known in the form of transition probabilities between the states, a Markov Decision Process (MDP) is an efficient tool to model and solve the problem. In this section, we formulate the AC-MPSD problem as an MDP; then, utilize Dynamic Programming (DP) to find the optimal solution via the Policy Iteration (PI) algorithm.

\subsection{MDP Formulation}
\label{sec:MDP}
A finite MDP is defined by the tuple $(\mathcal{S}, \mathcal{A}, \mathfrak{P}, \mathfrak{R},\gamma)$. $\mathcal{S}=\{1,2,\ldots, S\}$ is the set of the states of the environment. As for the action set, $\mathcal{A} = \{\mathcal{A}(1),\mathcal{A}(2),\ldots,\mathcal{A}(S)\}$, where $\mathcal{A}(s)$ is the set of the actions that the decision-making agent is allowed to take in state $s$. $\mathfrak{R}(s,a)\!\!: \mathcal{S} \times \mathcal{A}(s) \rightarrow \mathbb{R}$ is the reward function that determines the reward of each action $a \in \mathcal{A}(s)$ in state $s$. Function $\mathfrak{P}(s,a,s')\!\!: \mathcal{S} \times \mathcal{A}(s) \times \mathcal{S} \rightarrow [0,1]$ determines the probability of transition from state $s$ to state $s'$ when taking action $a$ in state $s$. Finally, $\gamma$ is the reward discount factor that is discussed in the following. To formulate the AC-MPSD problem as an MDP, the sets $\mathcal{S}$ and $\mathcal{A}$, the functions $\mathfrak{R}$ and $\mathfrak{P}$, and the parameter $\gamma$ are specified as follows.

\subsubsection{States}
\label{sec:States}
The state of the environment is defined as 
\begin{equation}
\label{eq:State}
s = ({\bm{C}^l}, {\bm{C}}_{\theta}^p, \bm{l}, \bm{f}, \bm{d}).
\end{equation} 
In this definition, $\bm{l} = [l_{1}, l_{2}, \ldots, l_{I}]$ and $\bm{f} = [f_{1}, f_{2}, \ldots, f_{I}]$ are, respectively, the numbers of currently deployed NSs in the CD and PD for each type of service. 
Borrowing ideas from \cite{wu2001admission}, $\bm{d}$ is the vector $[d_{1}, d_{2}, \ldots, d_{I}]$, where arrival (departure) of an NS of type $i$ is indicated by $d_{i} = +1$ $(d_{i} = -1)$. 
Note that since no simultaneous events occur at the same time, only one entry of $\bm{d}$ is non-zero. The reason behind definition \eqref{eq:State} is to maintain the Markov property and, also, to include sufficient details of the environment for computing the transition probabilities, which are discussed in the following.

\subsubsection{Actions and Rewards}
\label{sec:Actions}
Four actions/decisions are defined in the AC-MPSD problem. Action \textsf{accept} corresponds to the local deployment of the requested NS in the CD. To deploy the NS in the PD, the admission controller takes action \textsf{delegate}. The NS request is rejected with no profit or penalty and no resource consumption, if the action \textsf{reject} is taken. Moreover, a dummy action \textsf{none} is also defined, which is only taken when an NS is departing the network. This is an artificial action used to derive the transition probabilities in a tractable way as explained in the following subsection.

All actions are not allowed in every state. Let $\mathcal{S}^{+}_{i} = \{s \in \mathcal{S} \text{ s.t. } d_{i} = +1\}$, i.e., the set of states with an NS request arrival of type $i$, and $\mathcal{S}^{-}_{i} = \{s \in \mathcal{S} \text{ s.t. } d_{i} = -1\}$, i.e., the set of states with an NS departure of type $i$. The set $\mathcal{A}(s)$ determines the valid actions in state $s$ as follows:
\begin{equation*}
\label{eq:As}
\mathcal{A}(s) \text{ includes }
\begin{cases}
\mathsf{reject} & \text{if } s \in \mathcal{S}^{+}_{i} \\
\mathsf{accept} & \text{if } s \in \mathcal{S}^{+}_{i} \text{ and } \bm{c}_{i} \leq \bm{C}^{l}\\
\mathsf{delegate} & \text{if } s \in \mathcal{S}^{+}_{i} \text{ and } \bm{c}_{i} \leq \bm{C}^{p}_{\theta}\\
\mathsf{none} & \text{if } s \in \mathcal{S}^{-}_{i}
\end{cases}
\end{equation*}
\textsf{reject} is always in $\mathcal{A}(s)$ for $\mathcal{S}^{+}_{i}$, but \textsf{accept} (\textsf{delegate}) is included only if the CD (PD) has available resources.

In the AC-MPSD problem, the reward is the profit obtained by deploying each NS request $\delta_{i}$; so, it is independent of the next state and only determined by the action taken in the state. More specifically,  the rewards of actions \textsf{reject} and \textsf{none} are $\mathfrak{R}(s,\mathsf{none}) = \mathfrak{R}(s, \mathsf{reject}) = 0$; if the action \textsf{accept} is taken, the reward is $\mathfrak{R}(s,\mathsf{accept}) =  r_{i}$ and in the case of \textsf{delegate}, it is  $\mathfrak{R}(s,\mathsf{delegate}) =  r_{i} - \Delta(\delta_{i})$.

\subsubsection{Transition Probabilities}
\label{sec:Probabilities}
In AC-MPSD, the state transition probabilities are determined by the arrival and departure rates of NS requests. In this section, under the assumption of known $\lambda_{i}$ and $\mu_{i}$ $\forall i \in \mathcal{I}$, we obtain the transition probabilities $\mathfrak{P}(s,a,s')$ $\forall s,s' \in \mathcal{S}$ and $\forall a \in \mathcal{A}(s)$.

The transition from $s$ to $s'$ takes place in two stages. First,
the action $a$ taken in state $s$ is \textit{immediately} applied to the environment that changes domain resources $\bm{C}^{l}$ or $\bm{C}_{\theta}^{p}$ as well as $\bm{l}$ or $\bm{f}$; now, we say the system is in the \textit{transient} state $\tilde{s}=({\bm{C}^{l}}', {\bm{C}_{\theta}^{p}}', \bm{l}', \bm{f'}, -)$. Then, in the second stage, an arrival or departure event $\bm{d}'$ occurs that leads to the new state $s'=(\tilde{s},\bm{d}')$. These transitions are independent; therefore, the transition probability from $s$ to $s'$ is  $\mathfrak{P}(s,a,s') = \text{Pr}(\tilde{s}\, | \, s,a) \times \text{Pr}(\bm{d}'\, | \, \tilde{s})$. $\text{Pr}(\tilde{s} \, | \, s,a)$ is the probability of transition to state $\tilde{s}$ if action $a$ is taken in $s$ and it equals 1 for the \textsf{reject}, \textsf{accept}, and \textsf{delegate} actions, as the changes in the environment due to these actions are deterministic. However, for the \textsf{none} action, the departing NS $\delta_{i}$ can be from the CD or the PD. The probability of the former event is ${l_{i}}/({l_{i} + f_{i}})$, while the latter happens with probability ${f_{i}}/({l_{i} + f_{i}})$.  $\text{Pr}(\bm{d}' \, | \, \tilde{s})$ is the probability that the environment generates events $\bm{d}'$, i.e., arrival or departure of a new NS $\delta_{j}$ while in state $\tilde{s}$. Due to the Poisson assumption for the arrival and departure rates of the events, this probability is computed according to the competing exponentials theorem, hence the probability of an event is equal to the rate of the event divided by the total rates of all possible events. In our case, the total rate of events in state $\tilde{s}$ is $\Lambda(\tilde{s}) + M(\tilde{s})$ which are, respectively, the total arrival and departure rates. The details of function $\mathfrak{P}(s,a,s')$ are explained in Algorithm \ref{alg:Pr}, where $\bm{e}_{i}$ is a vector with $1$ in entry $i$ and $0$ elsewhere.

\begin{algorithm}[t]
\caption{$\mathfrak{P}(s,a,s')$}
\label{alg:Pr}
\begin{spacing}{1.3}
\begin{small}
\begin{algorithmic}[1]
    \If {$a \in \{\mathsf{reject}, \mathsf{accept}, \mathsf{delegate}\}$}
        \State $\text{Pr}(\tilde{s}\, | \, s,a) \gets 1$
    \Else
        \If {$\exists \, i \text{ s.t. } {\bm{l}}' = \bm{l} - \bm{e_{i}}$}
            \Comment{Departure from CD}
            \State $\text{Pr}(\tilde{s}\, | \, s,a) \gets \frac{l_{i}}{l_{i} + f_{i}}$
        \Else {$\ \exists \, i \text{ s.t. } {\bm{f}}' = \bm{f} - \bm{e_{i}}$}
        \Comment{Departure from PD}
            \State $\text{Pr}(\tilde{s}\, | \, s,a) \gets \frac{f_{i}}{l_{i} + f_{i}}$
        \EndIf
    \EndIf
    \State $\Lambda(\tilde{s}) \gets \sum_{i \in \mathcal{I}} \lambda_{i}$
    \Comment{Total arrival rate in $\tilde{s}$}
    \State $M(\tilde{s}) \gets \sum_{i \in \mathcal{I}} (l'_{i} + f'_{i})\mu_{i}$
    \Comment{Total departure rate in $\tilde{s}$}
    \If {$\exists \, j \text{ s.t. } s' \in \mathcal{S}^{+}_{j}$}
            \Comment{Arrival of type $j$; $d_{j}'=+1$}
            \vspace{0.8mm}
            \State $\text{Pr}(\bm{d}' \, | \, \tilde{s}) \gets \frac{\lambda_{j}}{\Lambda(\tilde{s}) + M(\tilde{s})}$
    \Else {$\ \exists \, j \text{ s.t. } s' \in \mathcal{S}^{-}_{j}$}
            \Comment{Departure of type $j$; $d_{j}'=-1$}
            \vspace{0.8mm}
            \State $\text{Pr}(\bm{d}' \, | \, \tilde{s}) \gets \frac{(l'_{j}+ f'_{j})\mu_{j}}{\Lambda(\tilde{s}) + M(\tilde{s})}$
    \EndIf
	\State \Return $\text{Pr}(\tilde{s}\, | \, s,a) \times \text{Pr}(\bm{d}' \, | \, \tilde{s})$ 
\end{algorithmic}
\end{small}
\end{spacing}
\end{algorithm}

\subsubsection{Discount Factor}
\label{sec:Discount}
Solving an MDP means finding a policy $\pi(s)$ that determines the action $a \in \mathcal{A}(s)$ $\forall s \in \mathcal{S}$ in order to maximize the cumulative reward obtained over time, which is called the expected return. In \emph{infinite} horizon MDPs, as it is the case in AC-MPSD, every policy $\pi$ with $\mathfrak{R}(s,a=\pi(s)) > 0$ will lead to the \emph{total} return 
$\sum_{t = 0}^{\infty} \mathfrak{R}(s_{{t}},a_{{t}}=\pi(s_{{t}})) = \infty$ 
regardless of the action taken in each state; hence, it does not make sense to compare the goodness of policies in this case.For this reason, commonly, the reward is \textit{discounted} \cite{sutton2018reinforcement, mahadevan1996average}; and the expected discounted return for each time $\bar{t}$, 
\begin{equation}
\label{eq:DiscountedReward}
G_{\bar{t}} = \sum_{t = \bar{t}}^{\infty} \gamma^{t}\mathfrak{R}(s_{t},a_{t}=\pi(s_{t})),
\end{equation}
is optimized, where $\gamma \in [0,1)$ is the discount factor. This discounting not only makes sure $G_{\bar{t}=0} \ll \infty$, but also it determines the importance of the immediate rewards compared to future rewards. For example, $\gamma = 0$ means that only the  immediate reward $\mathfrak{R}(s_{\bar{t}},a_{\bar{t}})$ is taken into account, which corresponds to the \textit{greedy} policy that does not consider the future rewards in making decisions. 

The objective of the AC-MPSD problem, defined in \eqref{eq:Profit}, is indeed maximizing the \textit{average} reward, not the discounted expected return. However, it is known that by setting $\gamma \rightarrow 1$, maximizing \eqref{eq:DiscountedReward} approximates the average reward \cite{mahadevan1996average}; so, in the MDP formulation of AC-MPSD, we set $\gamma \approx 1$.

\subsection{Policy Iteration Algorithm}
\label{sec:PI}
In this section, the optimal policy $\pi^{*}$ is obtained by solving the MDP using dynamic programming. For policy $\pi$, we define \textit{state-value} $v_{\pi}(s) = \mathbb{E}_{\pi}[ G_{0} | s_{0} = s]$, which is  
the expected discounted return starting from state $s$. The Bellman optimality equation \cite{sutton2018reinforcement} states that for the optimal policy $\pi^{*}$, we have 
\begin{equation*}
\label{eq:Bellman}
v_{\pi^{*}}(s) = \max_{a \in \mathcal{A}(s)} \sum_{s' \in \mathcal{N}(s,a)} \mathfrak{P}(s, a, s')\Big(\mathfrak{R}(s,a) + \gamma v_{\pi^{*}}(s')\Big);
\end{equation*}
where $\mathcal{N}(s,a)$ is the set of the possible next states in case of taking action $a$ in state $s$. Having the optimal state values, the optimal policy is $\pi^{*} = \argmax_{a} v_{\pi^{*}}$.

The recursive equation $v_{\pi^{*}}(s)$ can be solved by iterative dynamic programming methods such as the Policy Iteration algorithm \cite{sutton2018reinforcement}. The main loop of this algorithm is composed of two other loops. In the \textit{policy evaluation} loop, it evaluates the given policy $\pi$  by updating the state values $v(s)$ as
\begin{equation*}
v(s) = \sum\limits_{s' \in \mathcal{N}(s,a)} \mathfrak{P}(s,a,s')\Big(\mathfrak{R}(s,a)+\gamma v(s')\Big),
\end{equation*}
until the values converge. In the \textit{policy improvement} loop, for all the states, it updates the policy as
\begin{equation*}
\pi(s) = \argmax\limits_{a} \sum\limits_{s' \in \mathcal{N}(s,a)} \mathfrak{P}(s,a,s')\Big(\mathfrak{R}(s,a)+\gamma v(s')\Big).
\end{equation*}
The main loop terminates when there is not any change in the policy that implies the current policy is the optimal policy satisfying the Bellman equation.

To apply the PI algorithm on the AC-MPSD problem, besides the transition probabilities given by Algorithm \ref{alg:Pr}, the set $\mathcal{N}(s,a)$ is also needed for each action in each state, which is obtained by Algorithm \ref{alg:NextState}. It first finds the possible transient states $\tilde{s}_{1}$ (and $\tilde{s}_{2}$) according to the action $a$, then arrival (and departure) events are included to generate the next state $s'$.

Although the PI algorithm achieves the optimal solution, it can only be used for theoretical performance analysis rather than as a practical solution, because of the following unrealistic assumptions.
First, it needs the transition probabilities $\mathfrak{P}(s,a,s')$. However, the exact statistical information of the arrival/departure rates of the NS requests is typically not known. Second, it needs all the states of the MDP, but this is impractical, as the number of the states grows exponentially with the size of the problem including $I$, $\bar{C}^{l}_{k}/c_{i,k}$, $\bar{C}^{p}_{k}/c_{i,k}$ and $\theta_{k}$. Third, it assumes the environment \textit{immediately} transits from $s$ to $\tilde{s}$ before the next event occurs, i.e., the instantiation and termination of NSs take zero time; but, in practice, as elaborated in the experimental implementation in Section \ref{sec:experiments}, the actual time is not zero. Therefore, the real system is not exactly the MDP and, consequently, PI does not necessarily provide the optimal policy. In the next section, we present the RL-based solution that does not require these assumptions.

\begin{algorithm}[t]
\caption{NextState$(s,a)$}
\label{alg:NextState}
\begin{spacing}{1.15}
\begin{small}
\begin{algorithmic}[1]
    \State $i \gets $ the index of $\bm{d}$ which is not zero 
	\If {$a = \mathsf{reject}$}
	\Comment{No update in the domains}
	    \State $\tilde{s}_{1} \gets (\bm{C}^{l}, \bm{C}^{p}_{\theta}, \bm{l}, \bm{f}, -)$
    \ElsIf {$a = \mathsf{accept}$}
    \Comment{Update the CD}
        \State $\tilde{s}_{1} \gets (\bm{C}^{l} - \bm{c}_{i}, \bm{C}^{p}_{\theta}, \bm{l} + \bm{e}_{i}, \bm{f}, -)$
    \ElsIf {$a = \mathsf{delegate}$}
    \Comment{Update the PD}
        \State $\tilde{s}_{1} \gets (\bm{C}^{l}, \bm{C}^{p}_{\theta}- \bm{c}_{i}, \bm{l}, \bm{f} + \bm{e}_{i}, -)$
    \ElsIf {$a = \mathsf{none}$}
    \Comment{Either CD or PD can be updated}
        \State $\tilde{s}_{1} \gets (\bm{C}^{l} + \bm{c}_{i}, \bm{C}^{p}_{\theta}, \bm{l} - \bm{e}_{i}, \bm{f}, -)$
        \State $\tilde{s}_{2} \gets (\bm{C}^{l}, \bm{C}^{p}_{\theta} + \bm{c}_{i}, \bm{l}, \bm{f} - \bm{e}_{i}, -)$        
    \EndIf
	\For {$\tilde{s} \in \{\tilde{s}_{1}, \tilde{s}_{2}\}$}
	    \For {$j \in \mathcal{I}$}
	        \State $\bm{d}' \gets \bm{e}_{j}$
	        \Comment{Arrival event per-service type}
	        \State $s_{1}' \gets (\tilde{s}, \bm{d}')$
	        \If {$l'_{j} + f'_{j} > 0 $}
    	        \State $\bm{d}' \gets -\bm{e}_{j}$
    	        \Comment{Departure only if any deployed NS}
    	        \State $s_{2}' \gets (\tilde{s}, \bm{d}')$
	        \EndIf
	        \State $\mathcal{N} \gets \mathcal{N} \cup \{s_{1}', s_{2}'\}$
	    \EndFor
	\EndFor
	\State \Return $\mathcal{N}$
\end{algorithmic}
\end{small}	
\end{spacing}
\end{algorithm}

\section{Reinforcement Learning-Based Solutions}
\label{sec:Reinforce}
Reinforcement Learning (RL) is an alternative approach to solve MDPs where the decision-making agent \textit{learns} the optimal policy via interaction with the environment. In this section, first, we analyze the problem of applying the commonly-used Q-Learning algorithm for the AC-MPSD problem, and then, we present the R-Learning algorithm.

\subsection{Q-Learning Drawback}
\label{sec:QL}
Q-Learning is a widely used RL technique to solve sequential decision making problems. It works based on \textit{action-value} function $q_{\pi}(s,a)$, %
which is the expected discounted return starting from state $s$, performing action $a$, and then following policy $\pi$ \cite{sutton2018reinforcement}. The Bellman optimality equation states that 
\begin{equation*}
\label{eq:BellmanAction}
q^{*}(s,a) = \sum_{s' \in \mathcal{N}(s,a)} \mathfrak{P}(s, a, s')\Big(\mathfrak{R}(s,a) + \gamma \max_{a' \in \mathcal{A}(s')}q^{*}(s',a')\Big),
\end{equation*}
and, consequently, the optimal policy is
\begin{equation}
\label{eq:OptimalPolicyAction}
\pi^{*}(s) = \argmax_{a \in \mathcal{A}(s)} q^{*}(s,a).
\end{equation}

To iteratively solve this equation, Q-Learning maintains a table $Q[s,a]$ that estimates the action value and is updated by interacting with the environment as follows \cite{sutton2018reinforcement}:
\begin{equation*}
\label{eq:QUpdate}
Q[s,a] \gets Q[s,a] + \alpha \Big(\mathfrak{R}(s,a) + \gamma \max_{a' \in \mathcal{A}(s')} Q[s',a'] - Q[s,a]\Big).
\end{equation*}
This update is based on the \emph{bootstrapping} and \emph{temporal difference} concepts.
By bootstrapping, the agent assumes that the expected return in the next state $s'$ is $\max_{a' \in \mathcal{A}(s')} Q[s',a']$. Thus, it obtains a new estimate of $Q[s,a]$ as $\mathfrak{R}(s,a) + \gamma \max Q[s',a']$. Then, the temporal difference between the current value of $Q[s,a]$ and the new estimate is used to update the $Q$ table by the learning rate $\alpha$.

While Q-Learning is a popular RL algorithm and it has been applied successfully in a wide range of finite horizon episodic problems, it was shown that the algorithm cannot find the optimal solution in (some kinds of) infinite horizon MDPs due to maximizing the discounted reward  \cite{schwartz1993reinforcement}. The AC-MPSD problem is also an infinite horizon MDP, and the following theorem pinpoints the drawback of using Q-Learning for the problem. To the best of our knowledge, it is the first time of such a proof is provided for a non-artificial MDP.

\begin{theorem}
\label{theo:Gamma}
Let Pr$(a \, | \, s)$ be the probability of taking action $a$ in state $s$; define $f(\gamma)$ = Pr$(\mathsf{delegate} \, | \, s)$ $ - $ Pr$(\mathsf{accept} \, | \, s)$. First, $f(0) \leq 0$; second, $\exists \, \mathcal{S}' \subset \mathcal{S}$ where $\{\mathsf{delegate}$, $\mathsf{accept}\}$ $\in \mathcal{A}(s)$ 
$\forall s\in \mathcal{S}'$ and $f(\gamma)$ is an increasing function of $\gamma$.
\end{theorem}
\begin{proof}
The proof is given in the Appendix.
\end{proof}

We have the following corollaries from the theorem, which are also justified by the simulation results in Section \ref{sec:Simulation}:
\begin{itemize}
\item If $\gamma = 0$ then Pr$(\mathsf{delegate}\, | \, s) \ngtr $ Pr$(\mathsf{accept} \, | \, s)$, so the agent always prefers to deploy NSs in the CD rather than in the PD; i.e., it follows the sub-optimal \emph{greedy} policy. 
\item When the CD has sufficient resources, 
obviously the optimal policy is $\pi(s) = \mathsf{accept}$ $\forall s \in \mathcal{S}$, but $\gamma \approx 1$ implies the existence of  $\mathcal{S}' \subset \mathcal{S}$ such that Pr$(\mathsf{delegate} \, | \, s) > $ Pr$(\mathsf{accept} \, | \, s)$  $\forall s \in \mathcal{S}'$, leading to a sub-optimal policy.
\end{itemize}

Therefore, neither $\gamma \rightarrow 0$ nor $\gamma \rightarrow 1$ is the optimal setting for all configurations. In fact, as seen in the simulation results, the optimal value of $\gamma$ depends on the $\bar{\bm{C}}^{l}$, $\bar{\bm{C}}^{p}$ as well as on $\lambda_{i}$ and $\mu_{i}$ which are not known beforehand. As mentioned, the root of the problem is that Q-Learning finds the policy to maximize the discounted reward instead of the true objective of the AC-MPSD problem stated in Equation \eqref{eq:Profit}, that is \emph{average reward} maximization. In the next section, we use another reinforcement learning algorithm that directly optimizes the average reward.
 
\subsection{Average Reward RL}
\label{sec:RL}
An alternative solution to tackle the infinite accumulated reward issue is to maximize the average reward \eqref{eq:AverageReward} instead of the discounted reward \eqref{eq:DiscountedReward}:
\begin{equation}
\label{eq:AverageReward}
G_{\bar{t}} = \lim_{T \rightarrow \infty} \frac{1}{T}\sum_{t = \bar{t}}^{T} \mathfrak{R}(s_{t},a_{t}).
\end{equation}
To this end, define $T$-step state-value function for policy $\pi$ as  
\begin{equation}
\label{eq:TStepStateValue}
\tilde{v}_{\pi}^{T}(s) = \mathbb{E}_{\pi}\bigg[\sum_{t = 0}^{T}  \mathfrak{R}(s_{t},\pi(s_{t})) \  \Big | \ s_{0} = s\bigg],
\end{equation}
and the \emph{average return} of the policy $\pi$ as
\begin{equation}
\label{eq:AverageStateValue}
\rho_{\pi}(s) = \lim_{T \rightarrow \infty} \frac{\tilde{v}^{T}_{\pi}(s)}{T}.
\end{equation}
It is proved that in ergodic unichain MDPs, which is the case for the AC-MPSD problem, the average return for a given policy $\pi$ is independent of the state \cite{mahadevan1996average}, i.e., $\rho_{\pi}(s_{1}) = \rho_{\pi}(s_{2}) = \rho_{\pi}$ $\forall s_{1}, s_{2} \in \mathcal{S}$ where $\rho_{\pi}$ is the average return by policy $\pi$. This is the key observation in the development of the iterative algorithm, named R-Learning \cite{schwartz1993reinforcement}, to find the optimal policy for maximizing the average reward. 

Define $\rho^{*}$ as the maximum value of the average return of the MDP, which is not known in advance, and let $\rho$ be the current estimate of $\rho^{*}$ by the algorithm. In the R-Learning algorithm, through interactions with the environment in a number of episodes, the objective is to find a policy $\pi$ such that $\rho_{\pi} \approx \rho^{*}$. However, the problem of using \eqref{eq:TStepStateValue} to find the optimal policy is that it is possible to have two policies $\pi_{1}$ and $\pi_{2}$ such that  $\tilde{v}_{\pi_{1}}^{T}(s) > \tilde{v}_{\pi_{2}}^{T}(s)$ for a $T \ll \infty$ and some $s$.
In this case, naturally, $\pi_{1}$ should be preferred over $\pi_{2}$, but because of $\lim_{T \rightarrow \infty} \frac{1}{T}$ in definition \eqref{eq:AverageStateValue}, we have $\rho_{\pi_{1}} = \rho_{\pi_{2}}$ that implies no preference between the policies \cite{schwartz1993reinforcement, mahadevan1996average}. To solve this problem, in this context, the \textit{relative} (or bias) 
action-value function is defined as follows:
\begin{equation*}
\label{eq:ActionAverageValue2}
q_{\pi}(s,a) = \mathbb{E}_{\pi}\bigg[\lim_{T \rightarrow \infty} \sum_{t = 0}^{T} \Big(\mathfrak{R}(s_{t},\pi(s_{t})) - \rho_{\pi}\Big)\Big |\ s_{0} = s, a_{0} = a \bigg]
\end{equation*}
which can be seen as the relative gain of action $a$ in state $s$ compared to the average reward of the policy. By these definitions, for the \textit{bias-optimal} policy $\pi^{*}$ where $\rho^{*} \geq \rho_{\pi}$, we have  $q_{\pi^{*}}(s,a) \geq q_{\pi}(s,a)$ $\forall s \in \mathcal{S}$ $\forall a \in \mathcal{A}(s)$ \cite{mahadevan1996average}.
 
The R-Learning algorithm approximates the bias-optimal policy as follows. The algorithm starts from initial estimates of the average reward $\rho$ and the $Q$ table. Then, in $n$ episodes, each one composed of $m$ NS requests, it interacts with the environment and computes how much the reward of the action can be better than the estimated average return. 
More specifically, the agent, using an exploration strategy, such as $\epsilon$-greedy \cite{sutton2018reinforcement}, takes action $a$ in state $s$ and observes reward $\mathfrak{R}(s,a)$, then it computes an estimate of the action gain as 
\begin{equation*}
\label{eq:ActionAverageValue}
(\mathfrak{R}(s,a) - \rho) + \max_{a' \in \mathcal{A}(s')} Q[s',a'],
\end{equation*}
where the first term is the relative immediate gain of the action, and the second term, by bootstrapping, is the expected relative gain obtained in the future states following the same policy. Finally, similar to Q-Learning, the R-Learning algorithm also utilizes the temporal difference concept and updates the action value by a learning rate $\alpha$ as follows

\vspace{-2mm}
\small
\begin{equation*}
\label{eq:ActionAverageValueUpdate}
Q[s,a] \gets Q[s,a] + \alpha \Big((\mathfrak{R}(s,a) - \rho) + \max_{a' \in \mathcal{A}(s')} Q[s',a'] - Q[s,a]\Big).
\end{equation*}
\normalsize

Having $Q[s,a]$, the policy will be $\pi(s) \gets \text{argmax}_{a} Q[s,a]$ $\forall s \in \mathcal{S}$. But the $\rho$ used in this equation is not known in advance. Different approaches have been considered to learn it over time \cite{dewanto2020average}. Here, we use the following rule proposed in \cite{schwartz1993reinforcement} to update $\rho$ as
\begin{equation*}
\label{eq:AverageRewardUpdate}
\rho \gets \rho + \beta \Big(\mathfrak{R}(s,a) - \max_{\bar{a}} Q[s,\bar{a}] + \max_{a'} Q[s',a'] - \rho\Big),
\end{equation*}
where $\beta$ is the learning rate and $\mathfrak{R}(s,a) - \max_{\bar{a}} Q[s,\bar{a}] + \max_{a'} Q[s',a']$ is the new estimate of $\rho$ in case of taking action $a$ in state $s$. As explained in \cite{schwartz1993reinforcement}, to avoid the influence of the random actions by the exploration strategy, this update is only applied if action $a$ agrees with the policy. The details of the R-Learning algorithm are presented in Algorithm \ref{alg:RL}, where the hyperparameters $\alpha$, $\epsilon$, and $\beta$ are decayed with rate $\phi$ by the decaying function $\mathfrak{D}(\bar{x}, \phi) = \frac{\bar{x}}{1 + \phi \times i}$
at the beginning of episode $i$, and $Q[s, a]$ and $\rho$ are updated as discussed.

\begin{algorithm}[t]
\caption{R-Learning($n$, $m$, $\bar{\alpha}$, $\bar{\beta}$, $\bar{\epsilon}$, $\phi$)}
\label{alg:RL}
\begin{spacing}{1.05}
\begin{small}
\begin{algorithmic}[1]
	\State Arbitrary initialize $Q[s,a] \in \mathbb{R}$  $\forall s \in \mathcal{S}$, $\forall a \in \mathcal{A}(s)$
	\State $\rho \gets 0$
	\For {$n$ episodes}
		\State $\alpha \gets \mathfrak{D}(\bar{\alpha}, \phi)$,  $\epsilon \gets \mathfrak{D}(\bar{\epsilon}, \phi)$, $\beta \gets \mathfrak{D}(\bar{\beta}, \phi)$
		\State Reinitialize the environment
		\State $s \gets $ environment state $(\bar{\bm{C}}^{l},  \bar{\bm{C}}^{p}_{\theta}, \bm{0}, \bm{0}, \bm{d})$
		\For {$m$ NS requests}
			\State $a \gets $ action by an exploration strategy with parameter $\epsilon$
			\State Send action $a$ and the NS request to the environment
			\State Observe the outcome $s'$  and $\mathfrak{R}(s,a)$
			\State 
			{\footnotesize
			$Q[s,a] \gets Q[s,a] + \alpha \Big(\big(\mathfrak{R}(s,a) - \rho\big) +  \max\limits_{a'} Q[s',a'] - Q[s,a]\Big)$
			}
			\If {$Q[s,a] = \max\limits_{\bar{a}}Q(s,\bar{a})$}
				\State $\rho \gets \rho + \beta \Big(\mathfrak{R}(s,a) - \max\limits_{\bar{a}}Q[s,\bar{a}] + \max\limits_{a'}Q[s',a'] - \rho \Big)$
			\EndIf
			\vspace{-1mm}
			\State $s \gets s'$
		\EndFor
	\EndFor
	\State \Return $\pi \gets \argmax\limits_{a} Q[s,a]$ $\forall s \in \mathcal{S}$
\end{algorithmic}
\end{small}
\end{spacing}
\end{algorithm}

\section{Numerical Results}
\label{sec:Simulation}
In this section, after explaining the simulation setup, we investigate the performance of the dynamic programming, RL, and greedy algorithms via extensive simulations as well as through experimental assessment using the 5Growth platform.

\subsection{Simulation Setup}
\label{sec:SimSetup}
The default settings of the simulation parameters are summarized in Table \ref{tab:Simulation}, where to make the simulation scenarios more generic, we use term ``unit'' instead of specific metrics like CPU core, or \$. Moreover, in these simulations, we set $\omega_{i} = \omega$ $\forall i \in \mathcal{I}$ and $\theta_{k} = \theta$ $\forall k \in \mathcal{R}$. The performance of five practical algorithms are evaluated in comparison to the theoretical optimal solution obtained through dynamic programming (\textsf{PI}). These algorithms are the R-Learning, labeled as \textsf{RL}, Q-Learning with $\gamma = 0.20$, $\gamma = 0.55$, and $\gamma = 0.95$, which are respectively labeled as \textsf{QL-20}, \textsf{QL-55}, and \textsf{QL-95}; and also the greedy policy, labeled as \textsf{Greedy}. The greedy AC takes the \textsf{delegate} action only if there are not sufficient resources in the CD, and rejects NS requests only if there are not sufficient resources in the CD and PD. 

\begin{table} 
\begin{center}
\caption{Simulation parameter settings}
\vspace{-2mm}
\label{tab:Simulation}
\centering
\small\addtolength{\tabcolsep}{0pt}
\scalebox{0.9}{
\begin{tabular}{|c|c|}
  \hline
  Parameters & Values (unit) \\
  \hline
  \hline
  Number of resource and NS types: $R$, $I$ & 3, 3 \\
  \hline
  Consumer domain resources: $\bar{\bm{C}}^{l}$ & [30, 25, 30] \\
  \hline
  Provider domain resources: $\bar{\bm{C}}^{p}$ & [10, 15, 25] \\
  \hline
  NS type 1: $\bm{c}_{1}, r_{1}, \sigma_{1}$ & [4, 2, 1], 95, 80 \\
  \hline 
  NS type 2: $\bm{c}_{2}, r_{2}, \sigma_{2}$ & [2, 3, 2], 85, 40 \\
  \hline 
  NS type 3: $\bm{c}_{3}, r_{3}, \sigma_{3}$ & [2, 2, 4], 50, 5 \\
  \hline 
  NS 1 traffic rates: $\lambda_{1}, \mu_{1}$ & 10, 4 \\
  \hline
  NS 2 traffic rates: $\lambda_{2}, \mu_{2}$ & 11, 2 \\
  \hline
  NS 3 traffic rates: $\lambda_{3}, \mu_{3}$ & 12, 0.75 \\
  \hline
  Overcharging settings: $\theta$, $\omega$ & 2, 2 \\
  \hline
  Q-Learning hyperparameters: $\bar{\epsilon}$, $\bar{\alpha}$, $\phi$ & 1.0, 1.0, 0.025 \\
  \hline
  R-Learning hyperparameters: $\bar{\epsilon}$, $\bar{\alpha}$, $\bar{\beta}$, $\phi$ & 1.0, 1.0, 1.0, 0.025 \\
  \hline
  Learning parameters: $n,m$ & 2500, 4000 \\
  \hline
  \end{tabular}	
}		
\end{center}
\vspace{-3mm}
\end{table}

The overall evaluation procedure is as follows. In each experiment, the optimal policy is found through the PI algorithm. Then, Q-Learning and R-Learning are trained in $n$ episodes with $m$ random NS requests and the final policy is saved. Finally, a set $\mathcal{D}$ of $m$ NS requests is generated and the algorithms are applied to the set. This procedure is repeated 20 times for each setting and the average results are reported. 

In the following, two graphs are presented for each simulation. The first one is the optimality gap of each algorithm \textsf{Alg}, which is defined as: $Gap(\mathsf{Alg}) =\big(AP({\mathsf{PI}}) - AP({\mathsf{Alg}})\big)/AP({\mathsf{PI}})$, where $AP(\mathsf{Alg})$ is the average profit in Equation \eqref{eq:Profit} obtained by algorithm \textsf{Alg}. Moreover, to provide deeper insights on the operation of the algorithms, in each simulation, either the request acceptance rate $AR(\mathsf{Alg}) = |\mathcal{L_{\mathsf{Alg}}}| / |\mathcal{D}|$ or the Delegation rate $DR(\mathsf{Alg}) = |\mathcal{F_{\mathsf{Alg}}}| / |\mathcal{D}|$ 
is also reported, where $\mathcal{L}_{\mathsf{Alg}}$ and $\mathcal{F}_{\mathsf{Alg}}$ are the sets of the NSs deployed in the CD and the PD by algorithm \textsf{Alg}, respectively. 

\subsection{Learning Capability}
\label{sec:LearnCap}
In this section, we evaluate the performance of the RL algorithms to learn the optimal policy. To this end, Figure \ref{fig:EpisodeResults} compares the performance of the algorithms against the optimal solution with respect to the number of episodes $n$. The optimality gap shows that the RL algorithms are capable of \textit{learning} the optimal policy, as they approximate the optimal solution by increasing the number of episodes. However, the learning capability is different. R-Learning not only learns a better policy, but also achieves it sooner. Moreover, it can exploit the information more efficiently, i.e., while the gap of Q-Learning does not improve after 2000 episodes, the gap of \textsf{RL} continues to decrease when increasing $n$. 

The acceptance rates in Figure \ref{fig:EpisodeResults} show how the algorithms learn the policy. For small values of $n$, the algorithms do not explore the state space sufficiently, and consequently, there are a significant number of states for which the optimal decision is not found. By increasing the number of episodes, the algorithms discover more states wherein accepting the NS requests yields higher long-term profit. 

\begin{figure}[t]
\vspace{-3mm}
    \subfloat[Optimality Gap]{
        \hspace{-3mm}
        \includegraphics[width=0.5\linewidth]{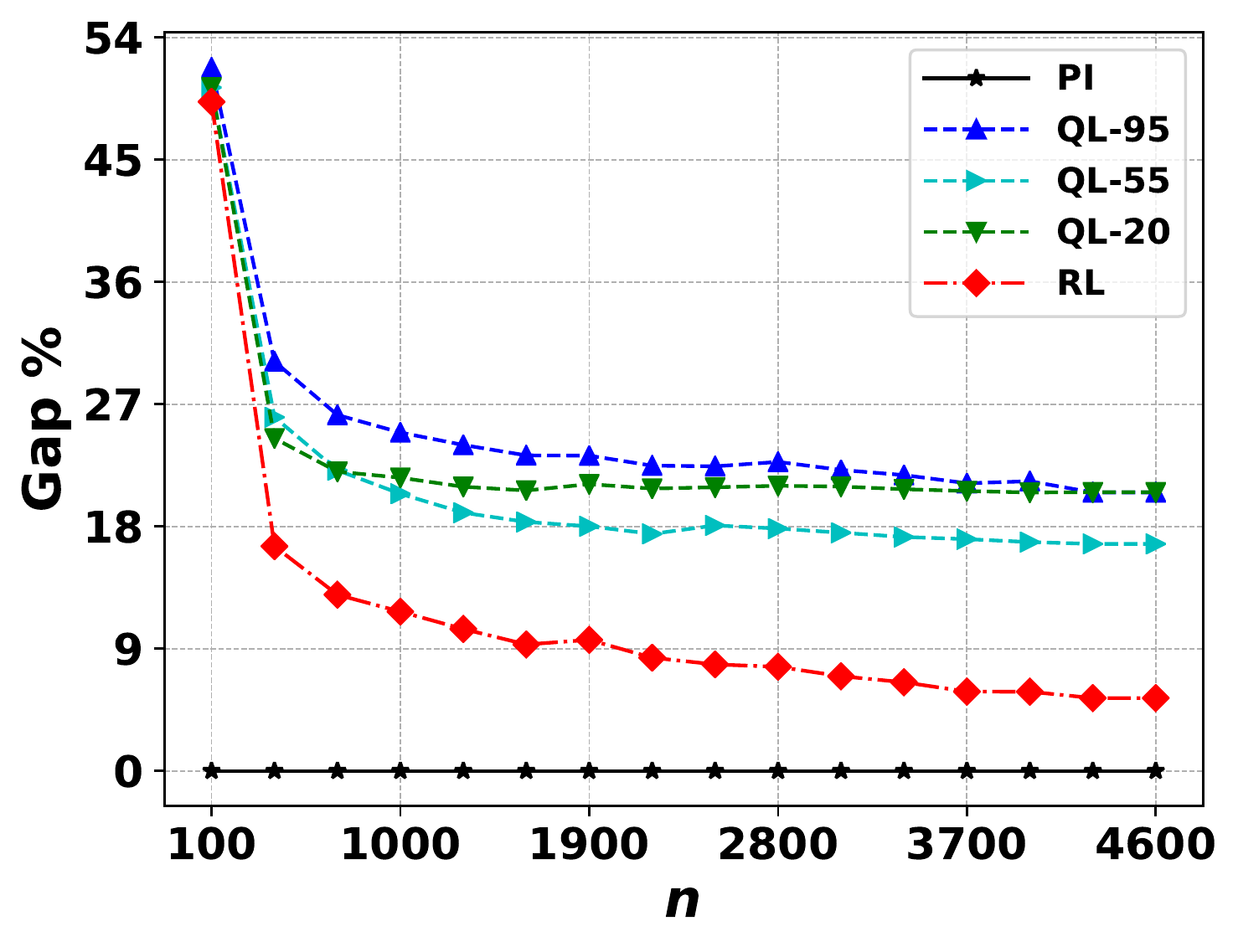}
    }
    \hfill
    \subfloat[Acceptance Rate]{
        \hspace{-4mm}
        \includegraphics[width=0.5\linewidth]{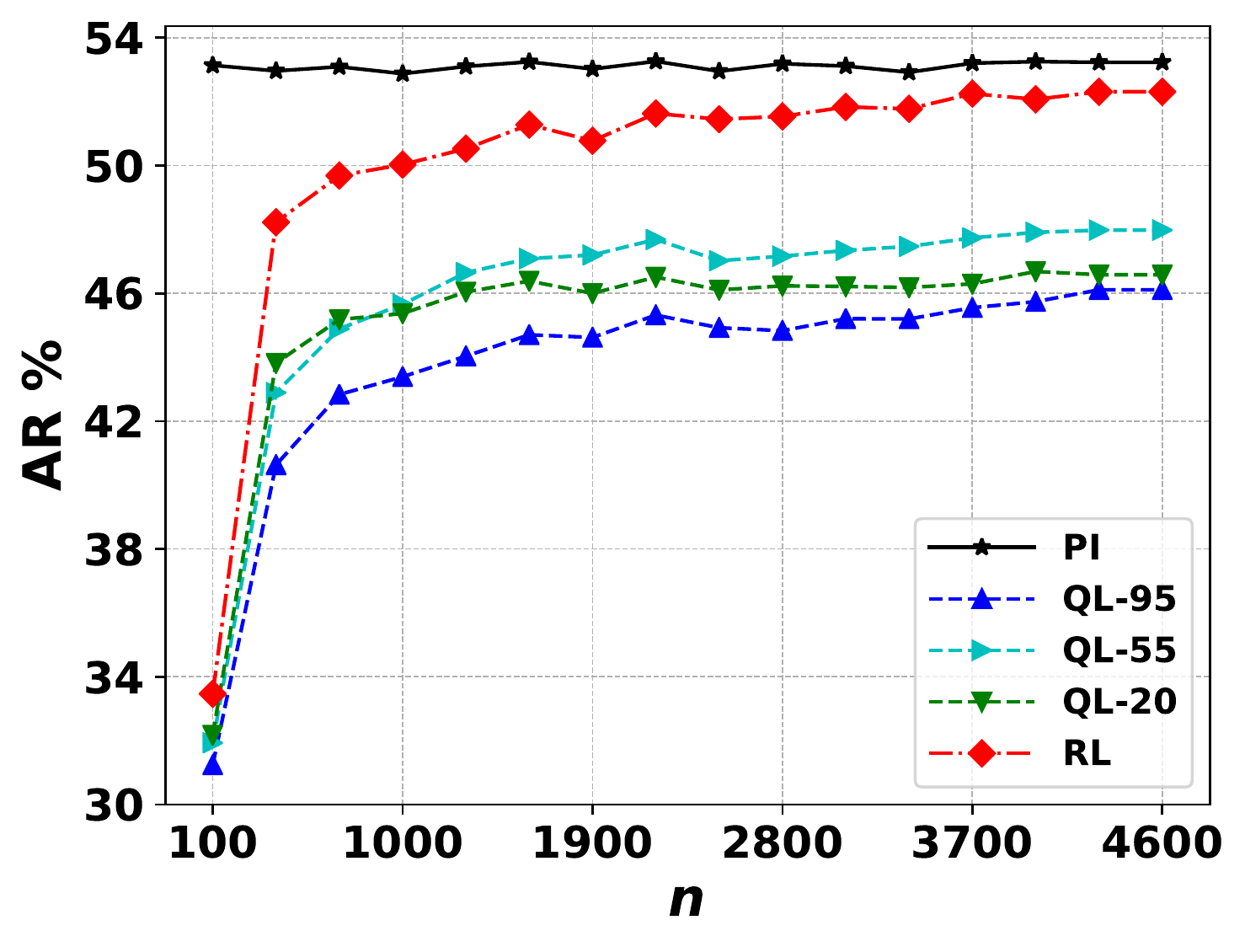}
        \hspace{-3mm}
    }
\vspace{-1mm}
\caption{Optimality gap and acceptance rate of the reinforcement learning algorithms with respect to the number of learning episodes.}
\label{fig:EpisodeResults}
\vspace{-3mm}
\end{figure}

\subsection{Resource Management Efficiency}
\vspace{-1mm}
\label{sec:ResourceEffect}
As discussed, AC is responsible for the management of domains' resources, as it determines the deployment domain for each NS request. In this section, we evaluate the performance of the algorithms in this respect. More specifically, 
the default values of the parameters $\bar{\bm{C}}^{l}$ and $\theta$ in Table \ref{tab:Simulation} are respectively replaced by $\eta_{C^{l}} \times \bar{\bm{C}^{l}}$ and $1 + \eta_{\theta}$, and the performance metrics are reported with respect to $\eta_{C^{l}}$ and $\eta_{\theta}$\footnote{In our simulations, the results of scaling $\bar{\bm{C}}^{p}$ are similar to the results of scaling the parameter $\theta$, which are omitted due to space limit.}.

\begin{figure}[t]
\vspace{-3mm}
    \subfloat[Optimality Gap]{
        \hspace{-3mm}
        \includegraphics[width=0.5\linewidth]{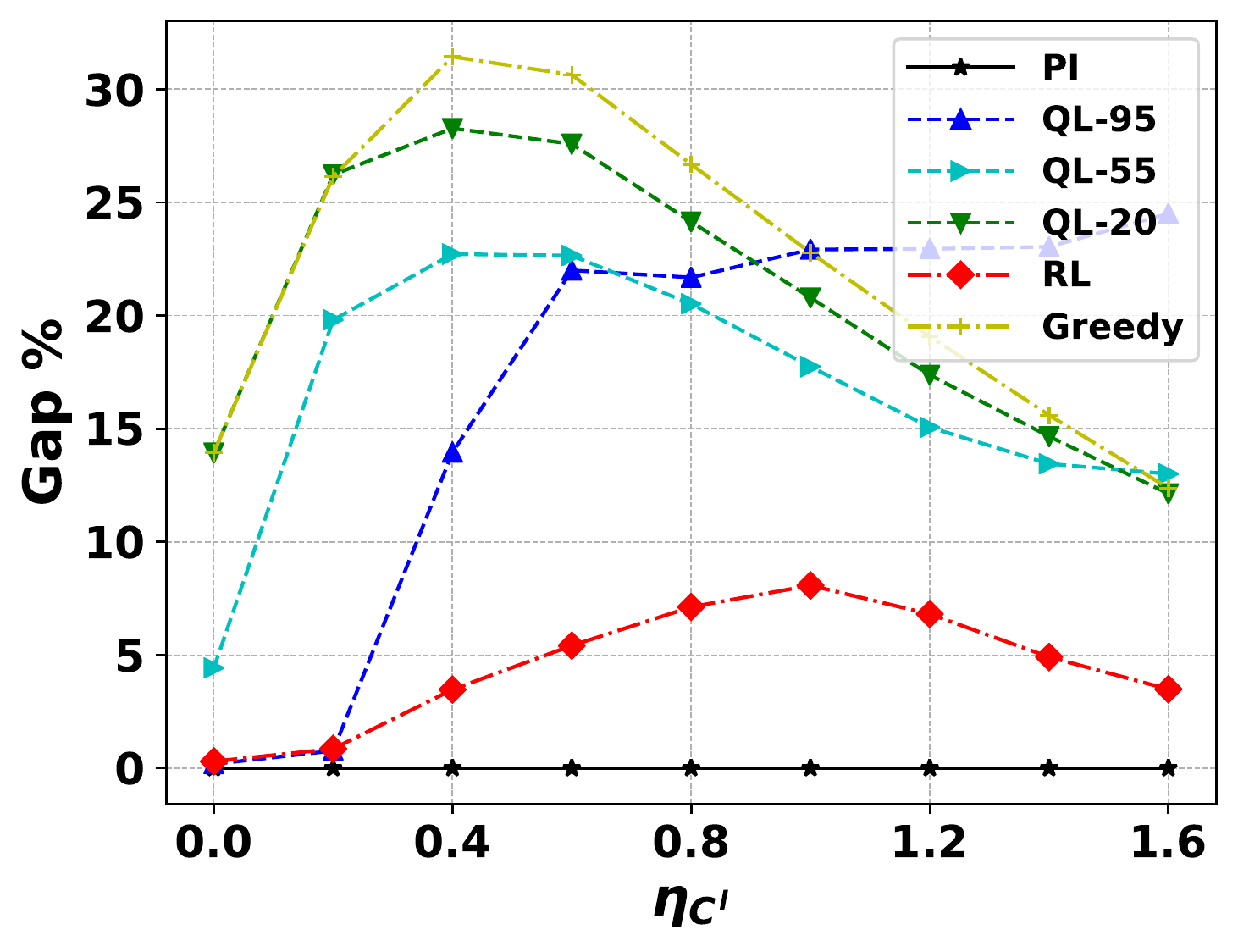}
        \label{fig:LocalCapacity:Gap}
    }
    \hfill
    \subfloat[Acceptance Rate]{
        \hspace{-4mm}
        \includegraphics[width=0.5\linewidth]{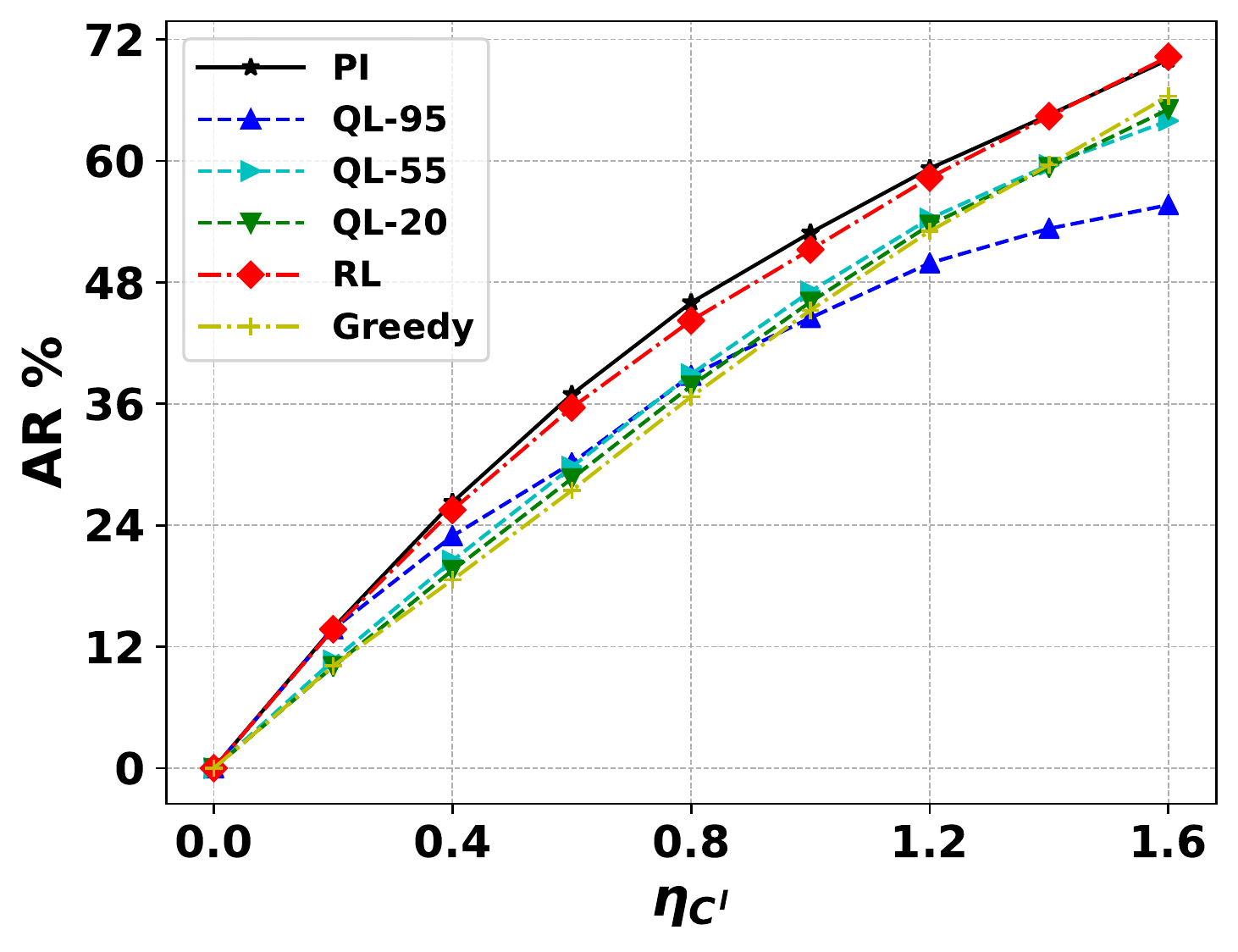}
        \label{fig:LocalCapacity:AR}
        \hspace{-3mm}
    }
\vspace{-1mm}    
\caption{Optimality gap and acceptance rate of the algorithms with respect to the scaling factor of CD's capacity.}
\label{fig:LocalCapacityResults}
\vspace{-1mm}
\end{figure}

The optimality gap in Figure \ref{fig:LocalCapacity:Gap} shows that R-Learning can efficiently utilize the resources; however, the performance of Q-Learning depends on the discount factor $\gamma$, and not a single value is the best setting for all the cases, which justifies the corollaries of Theorem \ref{theo:Gamma}. The optimality gap curves (except \textsf{QL-95}) are concave 
because in the case of very small $\eta_{C^{l}}$, the resources of the CD are scarce, so even the optimal policy by \textsf{PI} does not 
provide a considerably larger profit than other policies, i.e., all the policies are almost the same. In the case of very large $\eta_{C^{l}}$, there are sufficient resources in the CD, so  sub-optimal decisions by the practical algorithms do not yield a significant loss of average profit. Note that increasing the CD capacity by making $\eta_{C^{l}}$ bigger provides increases NS deployment opportunities in the consumer domain, which entails an increasing $AR$, as shown in Figure \ref{fig:LocalCapacity:AR}. \textsf{QL-95} does not use this opportunity because as stated by Theorem \ref{theo:Gamma}, it prefers \textsf{delegate} over \textsf{accept} due to the large $\gamma$.

The performance of the algorithms with respect to $\eta_{\theta}$ is shown in Figure \ref{fig:ThresholdResults}. Increasing $\eta_{\theta}$
represents more available resources in the PD, and consequently leads to a higher delegation rate, as shown in Figure \ref{fig:Threshold:FR}.
The (almost) constant optimality gap of \textsf{RL} shows that the algorithm skillfully manages PD resources considering the delegation (overcharged) costs.
The greedy policy cannot exploit the resources efficiently to improve the average profit. Again, the performance of Q-Learning depends on the discount factor $\gamma$.

\begin{figure}[t]
\vspace{-3mm}
    \subfloat[Optimality Gap]{
        \hspace{-3mm}
        \includegraphics[width=0.5\linewidth]{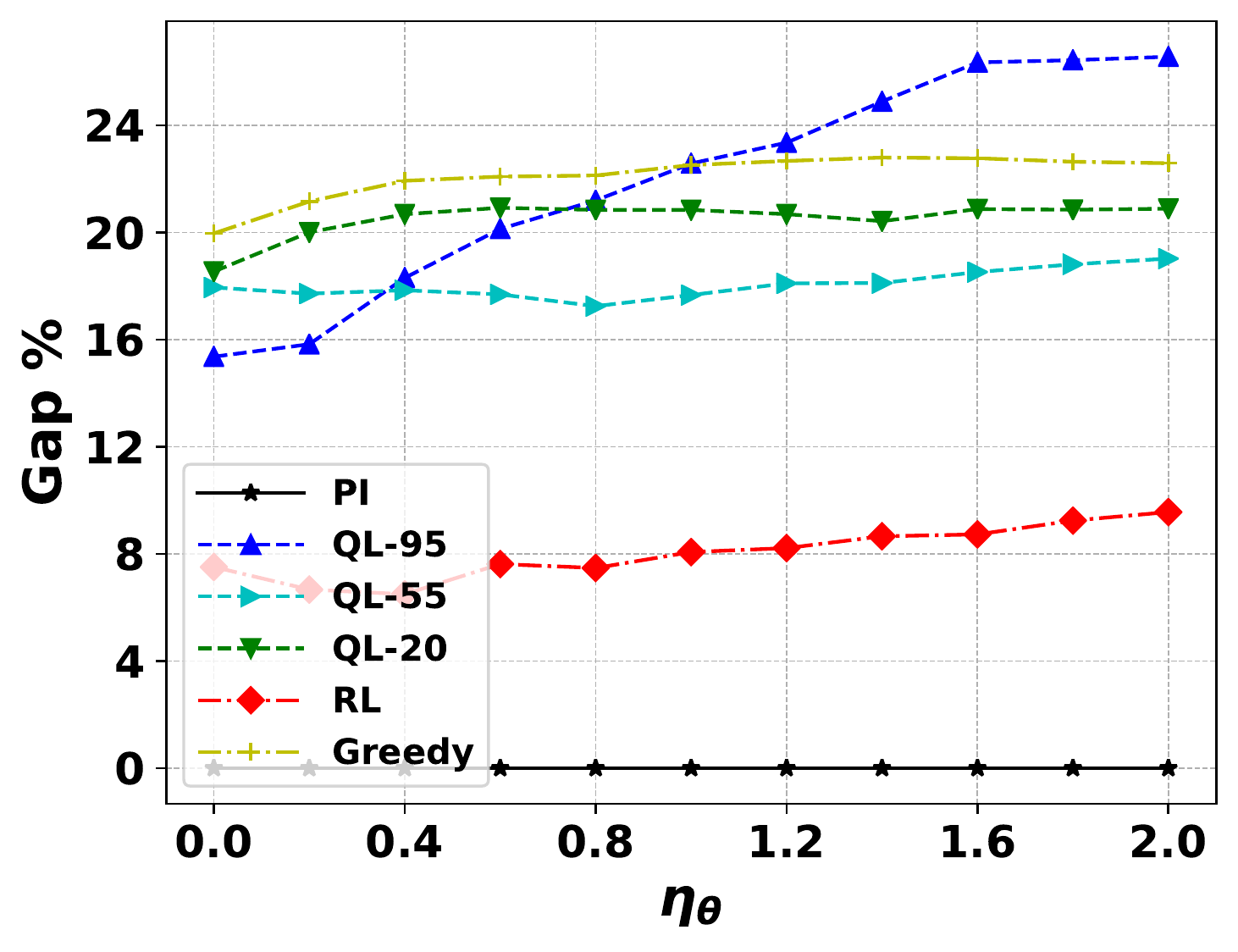}
    }
    \hfill
    \subfloat[Delegation Rate]{
        \hspace{-4mm}
        \includegraphics[width=0.5\linewidth]{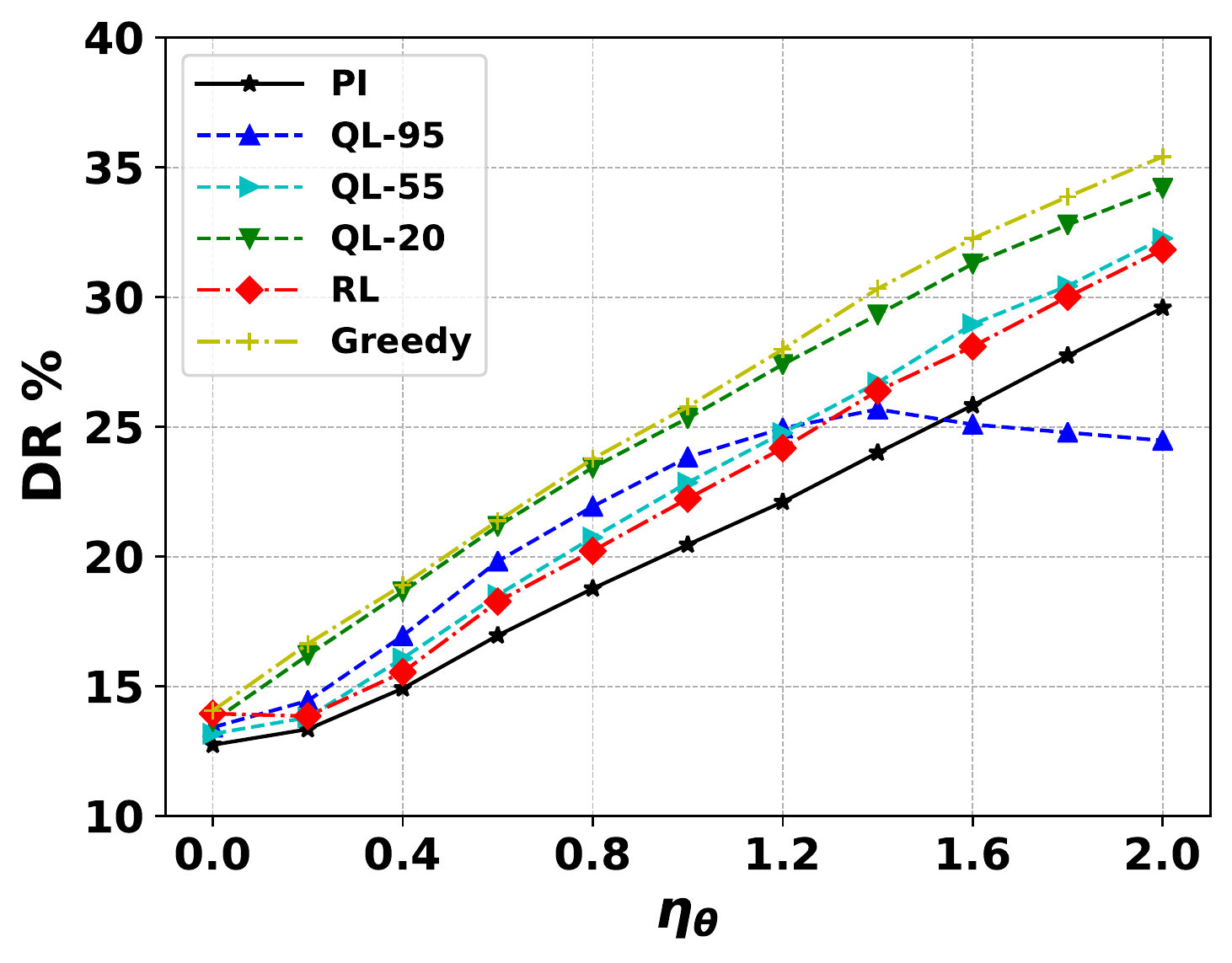}
        \label{fig:Threshold:FR}
        \hspace{-3mm}
    }
\vspace{-1mm}    
\caption{The  optimality gap and the delegation rate of the algorithms with respect to the reject threshold (${\theta_{k}} = 1 + \eta_{\theta}$).}
\label{fig:ThresholdResults}
\vspace{-3mm}
\end{figure}

\vspace{-3mm}
\subsection{Cost-Effectiveness}
\vspace{-1mm}
\label{sec:CostEffect}
In the AC-MPSD problem, the service delegation cost, which is determined by $\bm{\Sigma}$ and $\bm{\Omega}$, directly influences the average profit. In this section, we investigate the performance of the algorithms in this respect. To this end, the default value of $\bm{\Omega}$ in Table \ref{tab:Simulation} is scaled as $\eta_{\omega} \times \bm{\Omega}$ and the algorithms are evaluated with respect to it\footnote{In our simulations, the results of scaling the delegation fee $\bm{\Sigma}$ are similar to the scaling of $\bm{\Omega}$; they are omitted due to space limit.}. The results are shown in Figure \ref{fig:OverchargeResults}. Increasing $\eta_{\omega}$ incurs in higher service delegation costs, and consequently, decreases the delegation rate, as seen in Figure \ref{fig:Overcharge:FR}. Similar to the previous results, the optimality gap shows the superiority of R-Learning in taking the cost into account, as well as the dependency of Q-Learning performance on $\gamma$.

\begin{figure}[t]
\vspace{-3mm}
    \subfloat[Optimality Gap]{
        \hspace{-3mm}
        \includegraphics[width=0.5\linewidth]{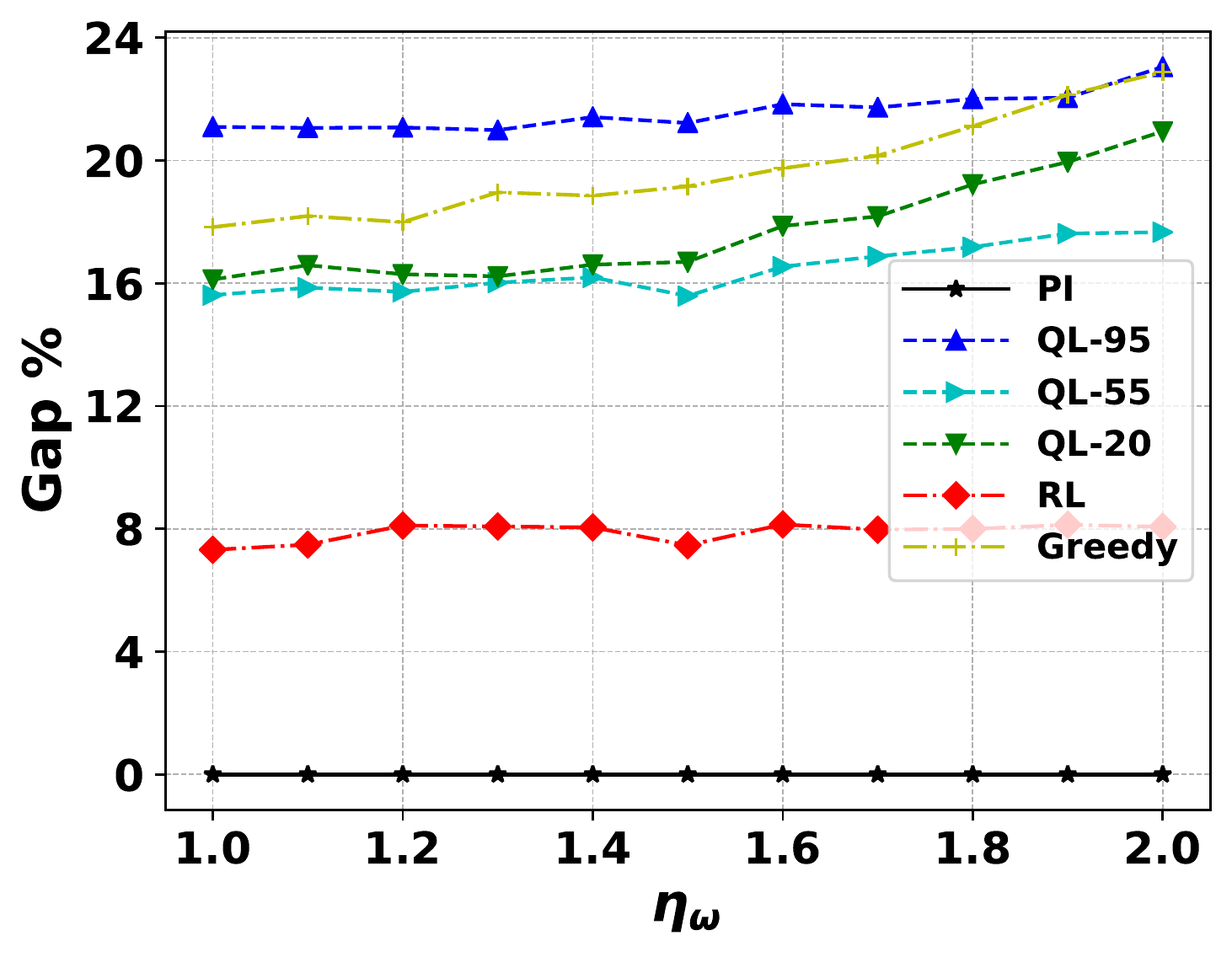}
    }
    \hfill
    \subfloat[Delegation Rate]{
        \hspace{-4mm}
        \includegraphics[width=0.5\linewidth]{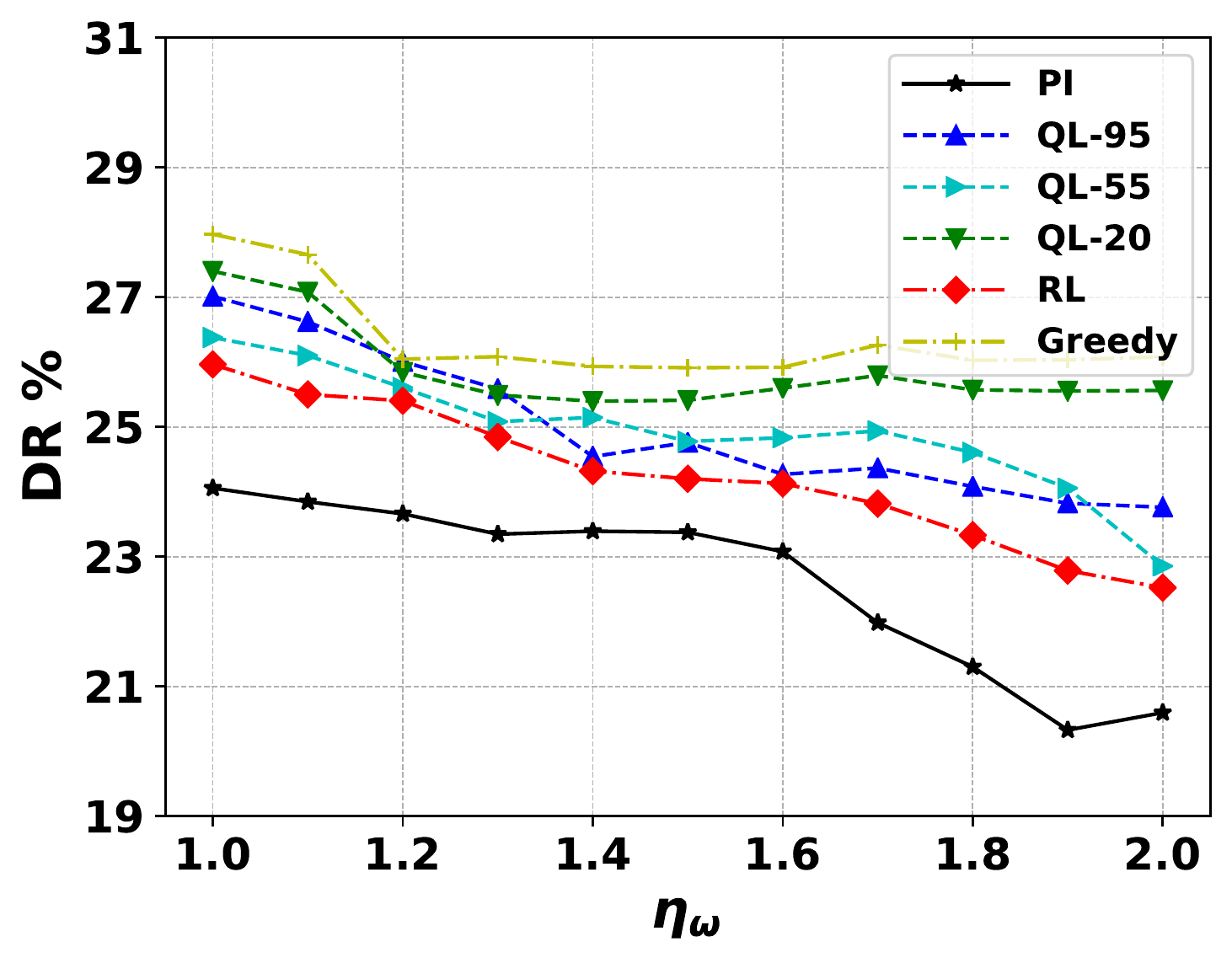}
        \label{fig:Overcharge:FR}
        \hspace{-3mm}
    }
\vspace{-1mm}    
\caption{Optimality gap and delegation rate of the algorithms with respect to the overcharging scaling factor (${\omega_{i}} = \eta_{\omega} \times {\omega}$).}
\label{fig:OverchargeResults}
\vspace{-1mm}
\end{figure}

The results presented in this section, show that R-Learning is a near-optimal solution that surpasses the other considered approaches; and also that there is not a fixed optimal value of $\gamma$ in the Q-Learning algorithm as stated by Theorem \ref{theo:Gamma}.

\vspace{-3mm}
\subsection{Experimental Evaluation Results}
\vspace{-1mm}
\label{sec:experiments}
This subsection presents a quantitative evaluation of the performance of the presented AC algorithms in a realistic testbed environment (EXTREME Testbed) using the 5Growth platform \cite{Xi_commag_2021}. The modular and flexible architecture of the 5Growth platform allows the straightforward integration of various AC techniques. 
Each AC policy, which is obtained by offline training (for \textsf{RL}) or computing (for \textsf{PI} and \textsf{Greedy}), is embedded as an external (containerized) module in the platform. It interacts through a well-defined REST API with the Service Orchestrator (5Gr-SO) of the platform, which is the architectural entity in charge of coordinating the end-to-end orchestration and lifecycle management of NSs in single- and multi-administrative domain scenarios. During the NS instantiation process, the Service Orchestration Engine (SOE) block of the 5Gr-SO contacts the AC module to determine the most suitable domain according to the AC policy (e.g., PI, RL, or Greedy). For that, according to the definition of the state in Equation \eqref{eq:State}, the SOE provides the resource characteristics of the NS to be instantiated, the number of active instances of each NS type in each domain, and the amount of available resources in each domain. It is worth mentioning that 
when establishing the federation contract, the CD and PD domains agree on the type of NSs they can delegate and the amount of reserved federation quota; and so, the SOE can easily derive the information required by the AC module from the 5Gr-SO ETSI NFV databases keeping track of the system status.

The setup used in this experimentation 
is composed of two interconnected domains running their own instance of the 5Growth platform and each domain having an underlying infrastructure consisting of an NFVI-Point of Presence (PoP). The NFVI-PoP of the CD has 12 CPU cores, 32GB of RAM, and 1TB of Storage, and the PD agrees to federate 6 CPU cores, 12 GB of RAM, and 500GB of storage with the CD. Due to the resource limitations in the testbed, 
for the experimental evaluations, we set $\bm{\Theta} = \bm{1}$, and we use two types of services, as shown in Table~\ref{tab:NS_chars}, where $\bm{c}_{i} = [$\#CPU, RAM (GB), Storage (GB)$]$.

\begin{table}[t] 
\begin{center}
\caption{Service types for experimental testbed evaluations}
\vspace{-2mm}
\label{tab:NS_chars}
\centering
\small\addtolength{\tabcolsep}{0pt}
\scalebox{0.9}{
\begin{tabular}{|c|c|c|c|c|c|}
\hline
$i$ & $\bm{c}_{i}$  & $\lambda_{i}$ (req/s) & $\mu_{i}$ (req/s) & $r_{i}$ & $\sigma_{i}$ \\ 

\hline
1 & [2, 2, 20] & $1/300$ & $1/800$ & 95 & 90 \\
\hline
2 & [1, 1, 15] & $1/345$ & $1/3000$ & 40 & 10 \\ 
\hline
\end{tabular}
}		
\end{center}
\vspace{-3mm}
\end{table}

Figure~\ref{fig:ExperimentResults} represents the total profit, $\sum_{i \in \mathcal{I}} \big(\sum_{\delta_{i} \in \mathcal{L}} r_{i} + \sum_{\delta_{i} \in \mathcal{F}} (r_{i} - \Delta(\delta_{i}))\big)$, obtained by the \textsf{PI}, \textsf{RL}, and \textsf{Greedy} solutions in ten independent experiments, each one covering the arrival and departure of NS requests
during a period of 5 hours, and also the average results. These results are consistent with the simulation results and show that the proposed average reward RL solution outperforms the greedy policy and provides near-optimal performance; i.e., in all the experiments, $\textsf{RL}$ obtains a higher total profit than the greedy approach. Thus, it shows that \textsf{RL} can efficiently learn the heterogeneity in network service types and select the appropriate deployment domain accordingly. 
\vspace{1mm}

However, contrarily to the simulation results, the gap between \textsf{RL} and \textsf{PI} is negligible, and in some experiments \textsf{RL} even outperforms \textsf{PI}. The detailed analysis of the traces of the experiments compared with the simulations, showed that the main reason is the \emph{non-zero} service instantiation and termination times. The MDP, the \textsf{PI} algorithm, and also the simulation environment are based on the assumption that the action taken by the agent is effective \emph{immediately}, i.e., before the next arrival/departure event is applied in the environment\footnote{Without this assumption, deriving the transition probabilities is not tractable, as it needs to consider (theoretically) infinite arrival/departure events while the environment is transiting from $s$ to $\tilde{s}$.}. This assumption is translated into \emph{zero} service instantiation and termination times, which does not correspond to real systems. For instance, in our experimental evaluations the time required to perform such lifecycle management operations ranges from 27 to 40 seconds. 
This implies that the states visited by the agent may not follow the probability distribution $\mathfrak{P}(s,a,s')$ obtained by Algorithm \ref{alg:Pr}; so, the policy by \textsf{PI}, which is based on the probabilities, is not necessarily optimal in the practice. Figure \ref{fig:ExperimentDiffs} shows the difference between the total profit of each policy in the experimental tests vs. the simulations using the same set of NS requests. As seen, the profit of \textsf{PI} in the testbed is always lower than the corresponding simulation, hence confirming the above analysis. Furthermore, the average performance loss of \textsf{RL} is less than that of \textsf{PI}, which represents another advantage of \textsf{RL} as a practical solution for AC-MPSD.

\begin{figure}[t]
\vspace{-3mm}
\begin{center}
    \includegraphics[width=0.95\linewidth]{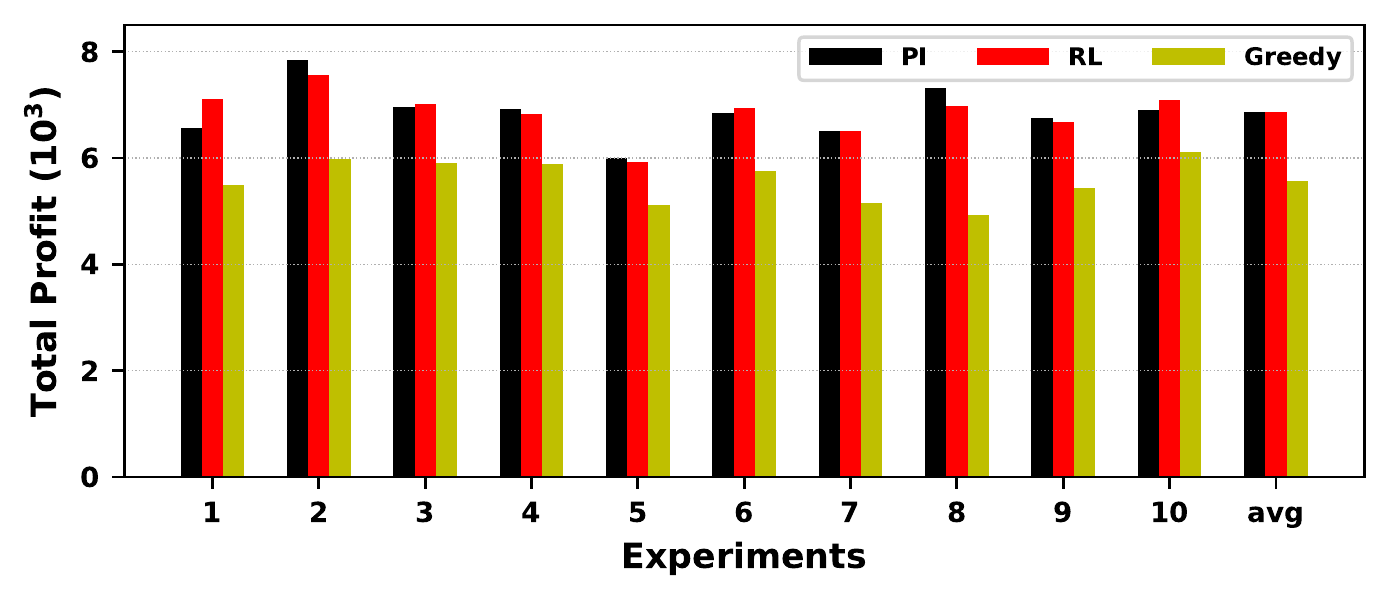}
\vspace{-5mm}
\caption{The total profit of the AC policies in the experiments in the testbed}
\label{fig:ExperimentResults}
\end{center}
\vspace{-3mm}
\end{figure}

\begin{figure}[t]
\vspace{-3mm}
\begin{center}
    \includegraphics[width=0.95\linewidth]{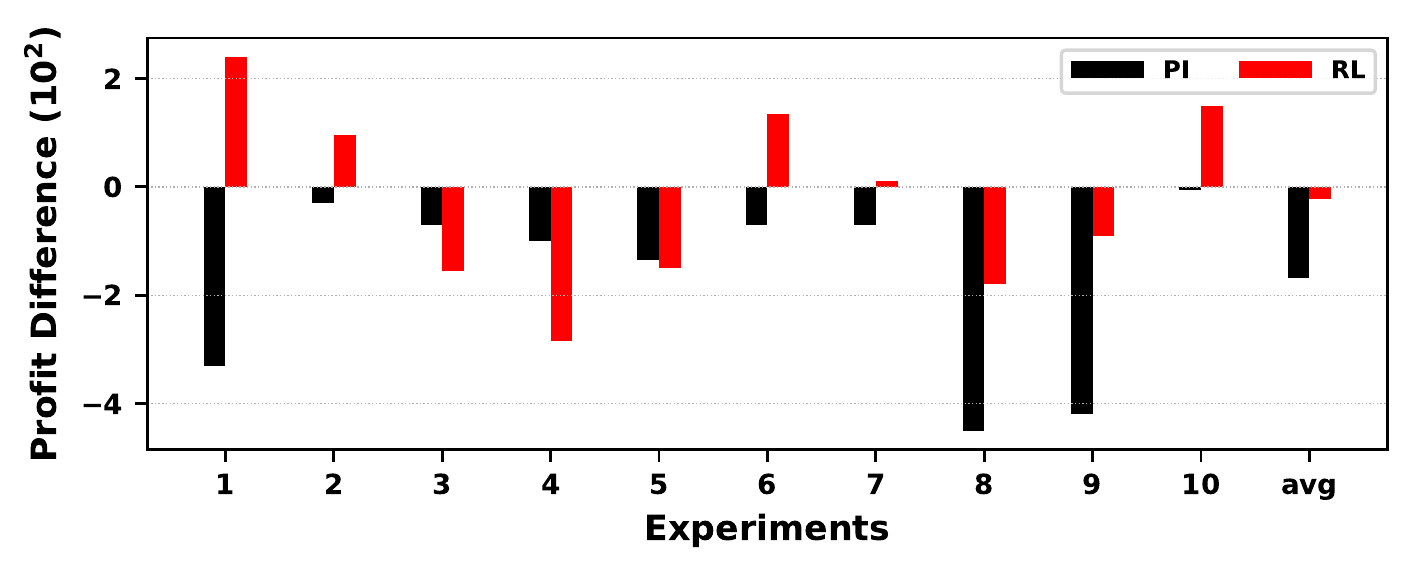}
\vspace{-5mm}    
\caption{Difference between total profit in experiments vs. simulations}
\label{fig:ExperimentDiffs}
\end{center}
\vspace{-3mm}
\end{figure}

\vspace{-1mm}
\section{Conclusions} 
\label{sec:Conclusion}
\vspace{-1mm}
We studied admission control for the multi-provider NFV service delegation problem, where the consumer domain 
can delegate provisioning of a service completely to the provider domain
subject to the federation contract. 
AC determines the deployment domain for each NS request, if it accepts the request. 
The theoretical optimal solution under ideal assumptions is obtained by modeling the problem as an MDP solved through 
the Policy Iteration algorithm. To tackle the problem in practice, where transition probabilities are not known and service lifecycle management operations take \emph{non-zero} time, we utilized RL. We showed, both analytically and via simulations, that the well-known Q-Learning algorithm, that optimizes the expected discounted return, is susceptible to the discount factor, whose optimal value 
cannot be determined in advance. We proposed the R-Learning algorithm that directly optimizes the \emph{average reward}. Experimental evaluations using the 5Growth platform as well as the simulation results showed R-Learning efficiently manages the resources of the domains, and skillfully considers the cost of delegation that leads to a near-optimal solution (with an optimality gap lower than 9\%) outperforming the Q-Learning and the greedy policies.
\vspace{-1mm}

\appendix
\vspace{-1mm}
In the exploration strategies used in Q-Learning, e.g., $\epsilon$-greedy, Pr$(a \, | \, s)$ is an increasing function of $Q[s,a]$; therefore, to prove the theorem, we need to show that Pr$(Q[s,\mathsf{delegate}] > Q[s, \mathsf{accept}]) \propto \gamma$. By the Q-Learning update equation, we have
\begin{eqnarray*}
Q[s,a] = \bar{Q}[s,a] + \alpha \Big(\mathfrak{R}(s,a) + \gamma \max_{a'} \bar{Q}[s',a'] - \bar{Q}[s,a]\Big),
\end{eqnarray*}
where for the sake of simplicity of discussion, the $Q$ values before updating are denoted by $\bar{Q}$ in the right hand side. 

Without loss of generality, assume that $\bar{Q}[s,\mathsf{delegate}] = \bar{Q}[s,\mathsf{accept}]$; and define $Q'[s,a] = \max_{a'} \bar{Q}[s',a']$, and 
\begin{align*}
\begin{split}
    g(\gamma) {}& = Q[s,\mathsf{delegate}] - Q[s,\mathsf{accept}].
\end{split}
\end{align*}

For the first part of the theorem where $\gamma = 0$, we have 
\begin{equation*}
 g(0) =  \alpha \Big(\mathfrak{R}(s,\mathsf{delegate}) - \mathfrak{R}(s,\mathsf{accept})\Big) = - \alpha \Delta(\delta) < 0,
\end{equation*}
that implies $Q[s,\mathsf{delegate}] < Q[s,\mathsf{accept}]$; and consequently $f(0) \leq 0$ that proves the first part.

For the second part, where $\gamma > 0$, we need to show $g(\gamma) > 0$ and $g(\gamma) \propto \gamma$. As it is seen,
\begin{align*}
\begin{split}
    g(\gamma) {} & = \alpha \Big(\mathfrak{R}(s,\mathsf{delegate}) - \mathfrak{R}(s,\mathsf{accept}) + \\ 
        & \ \ \ \ \gamma  (Q'[s,\mathsf{delegate}] -  Q'[s,\mathsf{accept}]) \Big) 
\end{split}
\end{align*}
so if $Q'[s,\mathsf{delegate}] \gg Q'[s,\mathsf{accept}]$, then the second term of $g(\gamma)$ is a large positive value, and consequently both conditions hold. It is easy to show that there are such states. Let 
$s = ({\bm{C}^l}, {\bm{C}}_{\theta}^p, \bm{l}, \bm{f}, \bm{e}_{i})$, so $Q'[s,\mathsf{delegate}]$ and $Q'[s,\mathsf{accept}]$ are respectively the expected discounted return starting from $s'_{\mathsf{delegate}}$ and  $s'_{\mathsf{accept}}$ where
\begin{equation*}
s'_{\mathsf{delegate}} = ({\bm{C}^l}, {\bm{C}}_{\theta}^p-\bm{c}_{i}, \bm{l}, \bm{f}+\bm{e}_{i}, \bm{d})
\end{equation*} 
\begin{equation*}
s'_{\mathsf{accept}} = ({\bm{C}^l}-\bm{c}_{i}, {\bm{C}}_{\theta}^p, \bm{l}+\bm{e}_{i}, \bm{f}, \bm{d})
\end{equation*} 

Now, assume there are a number of requests of type $\delta_{j}$ where $\bm{c}_{j} > {\bm{C}^l}-\bm{c}_{i}$ with very short life-time. None of them can be deployed in the CD in state $s'_{\mathsf{accept}}$ (and must be delegated) while they can be deployed in CD in state $s'_{\mathsf{delegate}}$. Therefore, we have
\begin{equation*}
Q'[s, \mathsf{accept}] = \sum_{t=0} \gamma^{t} (r_{j} - \Delta(\delta_{j}))
\end{equation*}
\begin{equation*}
Q'[s, \mathsf{delegate}] = \sum_{t=0} \gamma^{t} r_{j}
\end{equation*}
This implies that $Q'[s, \mathsf{delegate}] \gg Q'[s, \mathsf{accept}]$; and consequently $g(\gamma) > 0$ and $g(\gamma) \propto \gamma$, which proves the theorem.



\bibliographystyle{IEEEtran}    
\bibliography{main}

\end{document}